\documentclass[11pt,a4paper]{article}
\usepackage[margin=1in]{geometry}
\usepackage[T1]{fontenc}
\usepackage{lmodern,microtype}
\usepackage{amsmath,amssymb,amsthm,mathtools}
\usepackage{enumitem,booktabs,needspace,array}
\usepackage[numbers,sort&compress]{natbib}
\usepackage[hidelinks]{hyperref}
\usepackage{xcolor}
\usepackage{tikz,pgfplots}
\usetikzlibrary{arrows.meta,positioning}
\pgfplotsset{compat=1.18}
\newcommand{\Sch}{\mathcal S}
\newcommand{\PV}{\operatorname{PV}}
\newtheorem{theorem}{Theorem}[section]
\newtheorem{proposition}[theorem]{Proposition}
\newtheorem{conjecture}[theorem]{Conjecture}
\numberwithin{equation}{section}
\setlist{itemsep=3pt,topsep=5pt}
\hypersetup{pdftitle={Could four-dimensional phi4 theory exist beyond the usual triviality arguments?},pdfauthor={Anton Ragin}}
\title{Could four-dimensional $\phi^4$ theory exist\\beyond the usual triviality arguments?}
\author{Anton Ragin\\[0.3em]
\small Abdus Salam Centre for Theoretical Physics\\
\small Imperial College London, London SW7 2AZ, UK\\[0.3em]
\small\href{mailto:anton.ragin23@imperial.ac.uk}{\texttt{anton.ragin23@imperial.ac.uk}}}
\date{6 October 2026}
\begin{document}
\maketitle
\begin{abstract}
Do lattice triviality, suspected Borel non-summability and the
perturbative Landau pole exclude an interacting four-dimensional theory
satisfying the Osterwalder--Schrader axioms and reproducing every
renormalized coefficient of positive-coupling $O(N)$ $\phi^4$ theory?
An example involving one real function separates vanishing cutoff
limits along its allowed bare trajectories from a nonzero result obtained by renormalizing
the coefficients first; an exponentially small correction preserves its
expansion and removes the pole in its running coupling.
For the massless two-point expansion of the singlet composite
$\Phi_a\Phi_a$ at leading order in $1/N$,
ordinary Borel summation fails even for Gaussian smearing.
The arithmetic mean of two Borel integrals along opposite imaginary
rays defines a real distribution on every four-dimensional Schwartz
momentum or position test, with finite-order remainder bounds uniform
on bounded test sets.
Two alternative covariances with nonnegative mass spectral weights
remove the Euclidean pole and construct generalized free fields.
After a fixed contact subtraction, their Euclidean extensions preserve
every coefficient of the collective correlations obtained by Wick pairing.
The interacting finite-$N$ family remains open.
We also state a conditional interacting route through the physical
scalar correlations of a larger theory (the parent) with heavy particles. A concrete
$N=1$ gauge--Higgs--fermion candidate has a tuned one-loop trajectory
on which all interaction couplings tend to zero in the ultraviolet,
and a stable tree vacuum with additional masses comparable to the heavy
scale. Constructing the parent and its corrected ultraviolet trajectory,
and proving coefficient decoupling and the exact remainder bounds as
the quartic coupling tends to zero and the family's heavy masses grow,
remain open.
\end{abstract}
\clearpage
\setcounter{tocdepth}{1}
\begingroup\small\tableofcontents\endgroup
\clearpage

\section{The existence question}
\label{sec:question}

Does four-dimensional $\phi^4$ theory have to be free in the continuum?
The answer depends on what is required of a theory called
``$\phi^4$''. We ask for an interacting quantum field theory whose
correlation functions satisfy the Osterwalder--Schrader (OS) axioms---the
Euclidean consistency conditions for reconstructing a relativistic
quantum theory with positive state norms---and
have the complete renormalized perturbative expansion of the quartic
model. This definition specifies every perturbative coefficient, but
allows contributions that are smaller than every power of the coupling.
Throughout, \emph{interacting} means a non-Gaussian hierarchy of
elementary scalar correlations. Nontrivial scattering is a stronger
conclusion, requiring additional scattering assumptions and on-shell
control; it is not part of the criterion used here.
We conjecture that such a theory exists. The purpose of this article is
to explain why the usual triviality arguments do not, on their own,
exclude it, and to identify concrete construction mechanisms worth
investigating.

The end goal is a rigorous construction of those correlation functions.
Three familiar arguments are often offered against it: triviality
results give Gaussian continuum limits for specified classes of scalar
lattice models; positive-axis renormalon obstructions to Borel summation
have been proposed \cite{MagnenEtAl1987,Beneke1999}; and perturbative
running develops a Landau pole. Ordinary Borel summation for a positive
coupling forms
a transform by dividing the $k$th coefficient by $k!$, then integrates
its continuation along the positive real Borel axis against an
exponential kernel. Singularities on that axis or excessive growth
can obstruct the integral, but failure of this rule does not by itself
exclude exact correlations with the same asymptotic expansion.

Factorial bounds give a locally analytic Borel transform for massive
scalar $\phi^4_4$ at fixed external momenta
\cite{DeCalanRivasseau1981}; they do not establish
recovery along the entire positive ray. For comparison, high-order
minimal-subtraction RG data favour instanton-dominated coefficient growth
\cite{DunneMeynig2022}. Section~\ref{sec:collective} distinguishes
these results from the unresolved Borel summability of all finite-$N$
elementary correlation functions with the fixed physical normalization.
Each objection needs to be related to the definition above.

A lattice construction can first produce exact correlations and then
prove their perturbative expansion. Alternatively, one can try to
construct correlations from renormalized coefficients, with additional
non-perturbative information. Neither approach is built into our
definition.

If a theory with the required expansion is constructed, it need not be
the Borel sum of that expansion, and it cannot belong to a lattice class
whose limits are all Gaussian. The Landau-pole objection concerns a
divergence of the perturbative running coupling at a finite Euclidean
momentum scale. The corresponding exact response requires separate
existence and regularity checks.

\subsection{What must agree with quartic perturbation theory?}

Fix a finite number $N$ of real scalar components. The starting
prescription for the perturbative coefficients is the Euclidean bare
Lagrangian
\begin{equation}
 \mathcal L_{E,B}
 =\frac12(\partial\Phi_{B,a})^2
  +\frac12m_B^2\Phi_{B,a}^2
  +\frac{h_B}{8N}(\Phi_{B,a}\Phi_{B,a})^2.
 \label{eq:bare}
\end{equation}
Repeated component indices are summed. This expression supplies Feynman
rules and counterterms; it is not yet a definition of a continuum
measure or Hamiltonian.

Fix a renormalization prescription, including its convention for
contact terms at coincident insertions. Correlation functions are
distributions and can be singular when insertion points coincide.
Extending them to those points specifies local terms such as delta
functions supported at coincidence. Choose a positive particle pole mass
$m$ (the one-particle pole after continuation to timelike momentum),
a unit two-point pole residue,
and a fixed Euclidean momentum configuration for a quartic subtraction.
The configuration is nonexceptional: no nonempty proper subset of its
external momenta sums to zero. A chosen $O(N)$-invariant projection of
the connected amputated four-point function, with a fixed numerical
factor and sign, defines the dimensionless coupling $g$. We normalize
this projection so that $g=h_B/(32\pi^2)$ at tree level.\footnote{For
example, at a symmetric configuration one can use
$g=-\delta_{ab}\delta_{cd}G^{(4),\mathrm{amp}}_{abcd}/[32\pi^2(N+2)]$,
where $G^{(4),\mathrm{amp}}$ is the connected amputated response with
its momentum-conservation delta removed. Its tree value is
$-(h_B/N)(\delta_{ab}\delta_{cd}+\delta_{ac}\delta_{bd}+\delta_{ad}\delta_{bc})$.
The projection and configuration are fixed prescription data.} Keep
$m$, the residue condition and the subtraction momenta fixed as $g$
varies. The exact correlations we seek must satisfy these same
normalization conditions, not just their perturbative approximations.

The subtraction configuration has a fixed scale $\mu$. With the other
normalization data held fixed, $g$ is an input label for a theory,
which we denote by $\mathcal T_g$. Its point response at another scale
$Q$, when defined, is a different quantity:
\begin{equation}
 \bar g(Q;g):=\text{the same normalized four-point response in }
              \mathcal T_g\text{ at scale }Q,
 \qquad \bar g(\mu;g)=g.
 \label{eq:label-and-running}
\end{equation}
Varying $Q$ with $g$ fixed probes one theory; varying $g$ with $\mu$
fixed selects different theories. The expression ``running coupling''
refers to the former prediction, $\bar g(Q;g)$, or to using that
prediction as a new label at a different subtraction point wherever
that change of labels is locally invertible. In evaluating
$\bar g(Q;g)$, all external momenta are scaled by $Q/\mu$ from the
reference configuration, using the same projection and field
normalization. Tuning bare parameters as a cutoff is removed to keep
the physical normalization conditions fixed is a third operation,
discussed in Section~\ref{sec:cutoff}.

The exact normalization $\bar g(\mu;g)=g$ is required at the reference
point. Section~\ref{sec:poles} defines point responses at other scales
using shrinking smooth momentum packets and states sufficient
regularity conditions, including nonzero external propagator values.
The OS axioms and the fixed-test coefficient comparison below do not
by themselves establish these additional properties.

Renormalizing the Feynman expansion with the fixed physical conditions
gives coefficient distributions $a_{n,k}$ for each $n$-point function:
\begin{equation}
 \widehat S_n(g)=\sum_{k=0}^{\infty}g^k a_{n,k}.
 \label{eq:target-series}
\end{equation}
Here $n$ counts field insertions, $k$ counts powers of $g$, and component
indices are suppressed. The hat indicates a formal series: the
coefficients are specified without assuming convergence of the sum.
Renormalization removes the regulator from each coefficient separately
\cite{Zimmermann1969,EpsteinGlaser1973,Muller2003}.

Let $S_n(g)$ denote proposed exact Euclidean correlations. Agreement
with \eqref{eq:target-series} means that for every $n$, every integer
$K\geq1$, and every Schwartz test function $f$,
\begin{equation}
 \left\langle S_n(g)-\sum_{k=0}^{K-1}g^k a_{n,k},f\right\rangle
 =O_{n,K,f}(g^K),\qquad g\downarrow0.
 \label{eq:matching}
\end{equation}
The test function is smooth and rapidly decreasing in all insertion
coordinates; the brackets mean smearing the distribution with that
function. Thus agreement is required for every correlation function
and every perturbative order. The prescribed contact terms are
included in the coefficient comparison. Exact correlations may still
contain additional contact terms whose smeared contributions are
smaller than every power of $g$. This matters for the collective-field
example below.

In addition, the whole collection $\{S_n(g)\}$ must satisfy the OS
axioms, including reflection positivity, clustering and the growth
condition that permits reconstruction of tempered Minkowski fields
\cite[Section~IV]{OsterwalderSchrader1975}. The reconstruction test
domain excludes coincident insertions, whereas \eqref{eq:matching}
also fixes the perturbative coefficients of the contact terms.
Positivity is a requirement on correlations of different
field orders together; it cannot be
checked by accepting each $S_n$ independently.

\begin{conjecture}
\label{conj:existence}
For fixed finite $N$ and the massive physical normalization just
specified, an $O(N)$-covariant continuum family exists for all
sufficiently small $g>0$, satisfies the OS axioms and the exact
normalization conditions, and obeys \eqref{eq:matching}, including the
prescribed coefficient extensions at coincident points.
\end{conjecture}

Apart from the fixed light pole mass and residue, this criterion does
not prescribe the full mass spectrum at positive $g$. It permits
particle-production thresholds involving additional heavy states:
these are the timelike energies at which channels containing those
states open, typically giving a continuous spectrum or a branch cut.
Their masses may depend on $g$, provided their contributions preserve
\eqref{eq:matching} and the OS axioms. Section~\ref{sec:parent} explores
this possibility with heavy masses that grow faster than every inverse
power of $g$. Such masses are finite at each positive $g$. Choosing
these mass functions is additional non-perturbative input. Once that
choice and the normalization data are fixed, $g$ is the only variable
parameter of the scalar family $S_n(g)$.

Such a family is interacting: the quartic model has a nonzero
order-$g$ connected four-point coefficient at separated insertions.
Equation~\eqref{eq:matching} then makes that exact connected correlation
nonzero for sufficiently small $g>0$.

\subsection{Why non-perturbative terms enter the argument}

A contribution such as $e^{-1/g}$ is called \emph{flat} at $g=0$:
it is smaller than every power of $g$ as $g\downarrow0$. For a
distribution, we require this after each fixed Schwartz smearing.
Adding a flat term changes none of the coefficients in
\eqref{eq:matching}. That condition concerns each fixed test function;
it does not require uniform agreement at momentum scales that grow as
$g$ decreases. A correction can therefore be flat under fixed smearing
and still change the behavior near the perturbative Landau pole,
whose momentum scale tends to infinity as $g$ decreases.
We use \emph{completion} to mean
the choice of actual correlation functions with the prescribed
expansion, including any such non-perturbative terms. The axioms and
the physical normalization constrain this choice.

The first example is a real-valued function of momentum and coupling,
built from a regulated logarithmic series. The second uses the
massless leading-$1/N$ two-point expansion of a collective field formed
from the singlet composite $\Phi_a\Phi_a$, with its vacuum expectation
subtracted. Choosing the massless case isolates the logarithmic
summation and positivity problems. Bubble diagrams generate its logarithmic coefficients,
but the geometric sum's Euclidean pole still needs a distributional
prescription. At the leading order retained here, the collective
hierarchy is Gaussian: its higher coefficients follow from Wick
pairings of the two-point coefficients.
Completing the two-point function with nonnegative mass
spectral weights and forming higher correlations by Wick pairings
gives generalized free fields. Their connected functions above order
two vanish, although their spectra need not describe a single free
particle. With the stated contact convention they reproduce the
collective Wick coefficients. The conjecture requires the full
interacting finite-$N$ elementary hierarchy in \eqref{eq:target-series}.

Showing that these arguments leave the axiomatic existence question
open is not, by itself, evidence that the interacting family exists.
Section~\ref{sec:parent} adds a distinct conditional mechanism:
retaining the physical scalar correlations of a constructed larger
theory (the parent) would supply the whole positive hierarchy, while
coefficient and remainder bounds would give the quartic expansion.
One-loop and tree checks for a concrete
$N=1$ parent test the compatibility of the proposed ingredients and
motivate investigating it; the quantum construction and estimates
remain open.

The positive completions draw on earlier work on unphysical propagator
poles and analytic perturbation theory
\cite{Redmond1958,ShirkovSolovtsov1997,MiltonSolovtsov1997}.
Generalized free fields \cite{Greenberg1961}, distributional Borel
methods \cite{CalicetiGrecchiMaioli1986,CalicetiGrecchiMaioli1987Erratum},
and flat RG repairs \cite{KlaczynskiKreimer2014,MaiezzaVasquez2021}
likewise provide known tools. Here these tools bound the error after truncating the collective
two-point expansion on all four-dimensional momentum and position
Schwartz tests, using a radial momentum integral over $P^2\geq0$.
They also give two different positive completions with the same
collective coefficients.

Table~\ref{tab:objections} connects the three objections to the questions
examined here. Rigorous construction is the common requirement behind
them: the proposed correlations must exist and satisfy the axioms
together.
\begin{table}[htbp]
\centering
\small
\caption{What each usual objection leaves to be decided for our existence question.}
\label{tab:objections}
\begin{tabular}{@{}p{.20\linewidth}p{.37\linewidth}p{.37\linewidth}@{}}
\toprule
Objection & Mechanism examined here & Remaining question\\
\midrule
Gaussian lattice limits & In the function example, taking cutoff limits
and reconstructing a function from renormalized coefficients give
different results. & Must every theory with the quartic expansion
arise from a lattice class whose limits are provably Gaussian?\\[3pt]
Borel failure & A distributional prescription and flat correction can
complete the collective two-point expansion. & Can one complete the interacting
hierarchy, with all axioms satisfied?\\[3pt]
Landau pole & A flat correction can remove a running-coupling pole
while preserving its perturbative expansion. & Do the exact correlations
admit regular point responses at finite Euclidean scales?\\
\bottomrule
\end{tabular}
\end{table}

The argument proceeds in that order. Section~\ref{sec:cutoff} examines lattice
triviality through the logarithmic function example, then removes its pole by a flat
correction to the renormalization-group flow. Section~\ref{sec:collective} derives the logarithmic collective
two-point coefficients from bubble diagrams, explains why their known geometric
sum still needs a distributional prescription, and proves that the summed distribution has the prescribed expansion
after smearing with every four-dimensional momentum or position Schwartz
test (Theorem~\ref{thm:bf}). It also constructs two positive completions
(Proposition~\ref{prop:completions}). Section~\ref{sec:poles}
distinguishes singularities of a coupling used to label a theory,
allowed timelike poles and Euclidean four-point poles excluded by the
stated regularity assumptions. Section~\ref{sec:parent} addresses interaction
by retaining the scalar correlations of a larger theory; a conditional
theorem states the heavy-mass estimates that would give the quartic
expansion. Section~\ref{sec:next} identifies the remaining construction tests and
the physical information needed to select a completion.
Appendices~\ref{app:distributional}--\ref{app:gaussian} give the
summation estimates, spectral checks and Wick-hierarchy arguments;
Appendices~\ref{app:massive}--\ref{app:parent} treat the massive saddle,
a timelike bound-state response and the parent candidate.

\section{Lattice limits and the order of construction}
\label{sec:cutoff}

The first objection concerns a particular way to construct a continuum
theory: start with scalar lattice systems and remove their cutoff.
If their limits are Gaussian, this route cannot produce the interacting
family in Conjecture~\ref{conj:existence}. To exclude that family
altogether, one needs a further statement: every OS theory with the
prescribed quartic expansion can be obtained through such a route.
We examine why that statement is additional content, using a function
with two explicitly calculable construction procedures. The first sums
at finite cutoff, then removes the cutoff while tuning a nonnegative
bare parameter. The second renormalizes and removes the cutoff from
each formal coefficient, then sums. Their different outcomes identify
the connection that a general lattice nonexistence inference needs.

During cutoff removal, one tunes $h_B(\Lambda;g)$, when possible,
to keep $g$ and the other physical normalization conditions fixed.
This is sometimes called running of the bare coupling. Energy running
$\bar g(Q;g)$ in \eqref{eq:label-and-running} instead varies $Q$
within one fixed continuum theory.

\subsection{What the triviality results rule out}

Aizenman and Duminil-Copin prove Gaussianity of the four-dimensional
scaling limits covered by their theorem for ferromagnetic Ising-type
systems and the associated one-component lattice quartic fields,
under the hypotheses of their theorem and corrigendum
\cite{AizenmanDuminilCopin2021,AizenmanDuminilCopin2024Corrigendum}.
Other rigorous results treat further scalar classes and weak-coupling
regimes of $O(N)$ models
\cite{Panis2026,BauerschmidtBrydgesSlade2014}. Together with the lattice
and RG analyses of L\"uscher and Weisz
\cite{LuscherWeisz1987,LuscherWeisz1989}, these give substantial support
to the conventional triviality picture. Tuning bare parameters within
a class proved to have Gaussian limits
cannot produce an interacting limit.

The existence criterion \eqref{eq:matching} fixes the renormalized
coefficients, but does not require a representation by a specified
scalar lattice measure. Inferring nonexistence therefore requires a
theorem that every family meeting this criterion has such a
representation, or another argument ruling out all remaining families.
Starting with renormalized coefficients is a legitimate constructive
possibility: Rivasseau's planar four-dimensional matrix construction
provides a precedent on the asymptotically free, negative-coupling side
\cite{Rivasseau1984Planar}. That result does not settle the finite-$N$
vector problem here. Its relevance is that coefficientwise
renormalization can precede summation in an actual construction.

\subsection{The same formal coefficients, two different outcomes}

The example isolates the two orders of construction for a real
function; no field axioms are imposed here.
Let $x=P^2/\mu^2>0$ be a squared probing momentum divided by a fixed
reference value, and let $L=\Lambda^2/\mu^2$ be the squared cutoff
ratio. Thus changing $x$ probes a different momentum, while changing
$L$ changes the regulator. The reference point is $x=1$.
Restrict the bare parameter to $g_0\geq0$, mimicking the stability
restriction on the coupling of a scalar cutoff action with a quartic
interaction. This domain restriction is part of the function example.
Start from the formal series
\begin{equation}
 \widehat Y_L(x;g_0)
 =\sum_{j\geq0}[-\log(L/x)]^j g_0^{j+1}.
 \label{eq:toy-bare}
\end{equation}
At fixed $L>x$, the geometric series converges when
$g_0\log(L/x)<1$ and has the following rational sum. The displayed
Borel--Laplace integral extends that sum to every $g_0>0$, including
values outside the convergence interval: its integrand decays
exponentially because $L>x$.
\begin{equation}
 \begin{aligned}
 Y_L(x;g_0)&=\int_0^\infty e^{-\xi/g_0}(x/L)^\xi\,d\xi
          =\frac{g_0}{1+g_0\log(L/x)},\quad g_0>0,\\
 0&\leq Y_L(x;g_0)\leq\frac1{\log(L/x)},\quad g_0\geq0,
 \end{aligned}
 \label{eq:toy-bound}
\end{equation}
Here $\xi$ is the Borel variable. For the series starting at $g_0$, we
use $\int_0^\infty e^{-\xi/g_0}\xi^j/j!\,d\xi=g_0^{j+1}$;
no $1/g_0$ prefactor is needed. Expanding $(x/L)^\xi$ therefore
recovers \eqref{eq:toy-bare} within its convergence interval.
Set $Y_L(x;0)=0$. The bound holds however large $g_0$ is. Consequently
$Y_L(x;g_0(L))\to0$ as $L\to\infty$ along every allowed bare
trajectory, at each fixed $x$. The convergence is uniform on fixed
compact subsets of $x>0$, and also holds after compactly supported
smooth smearing. Summing the regulated series first therefore gives
a zero cutoff limit, regardless of how the numerical bare parameter
is tuned within its allowed domain.

Now renormalize the formal series before summing it. Define $g$ by
the formal subtraction condition $\widehat Y_L(1;g_0)=g$, then express
$g_0$ as a formal series in $g$. The renormalized reference value $g$
will be held fixed as $x$ varies; $g_0$ is the cutoff parameter.
The subtraction relation and its formal inverse are
\begin{equation}
 g=\sum_{j\geq0}(-\log L)^j g_0^{j+1},\qquad
 \widehat g_0(g;L)=\sum_{j\geq0}(\log L)^j g^{j+1}.
 \label{eq:toy-inverse}
\end{equation}
The hat on $\widehat g_0$ marks the formal inverse. Its rational
continuation would assign the numerical value $g_0=g/(1-g\log L)$,
which is in the permitted nonnegative domain only for $g\log L<1$.
At $x=1$, the bound in \eqref{eq:toy-bound} forces every permitted
numerical reference value to tend to zero. Formal inversion is still
possible at each finite $L$: substitution and collection of equal
powers of $g$ give
\begin{equation}
 \widehat Y_L(x;\widehat g_0(g;L))=\sum_{j\geq0}(\log x)^j g^{j+1}.
 \label{eq:toy-renormalized}
\end{equation}
Every coefficient is now independent of $L$, so removing the cutoff
separately from each coefficient leaves the series unchanged. Only
after those operations do we sum and choose a numerical input $g>0$,
obtaining
\begin{equation}
 q(x,g)=\frac{g}{1-g\log x},\qquad q(1,g)=g.
 \label{eq:toy-q}
\end{equation}
The geometric series converges when $|g\log x|<1$; its positive-axis
Borel integral gives the same expression for $\log x<1/g$.
Elsewhere we use the displayed rational continuation, away from
its pole. For fixed $g$, $q(x,g)$ is the prediction as the probing
momentum $x$ varies; changing $g$ selects another member of this
function family. The second procedure produces a nonzero function
from the renormalized coefficients without requiring a nonnegative
numerical bare trajectory reaching it.

There is a precise reason for the discrepancy. Numerically holding
$g>0$ fixed in \eqref{eq:toy-inverse} requires a bare parameter that
diverges at $L=e^{1/g}$ and becomes negative beyond it. The permitted
bare domain cannot maintain that reference value to arbitrarily large
cutoff. The formal inverse still exists, since each of its coefficients
is finite at each finite $L$. At cutoffs where the numerical inverse
is permitted, the two summed functions agree exactly. They differ as
complete construction procedures because the allowed bare trajectories
cannot reach the second result.

The example shows that a vanishing limit within an allowed bare domain
does not, as a general principle, force the renormalized coefficients
to sum to zero. It makes no claim about an interchange in the
Aizenman--Duminil-Copin proof. The further question for quartic theory
is whether every OS family with the prescribed expansion must arise
from a cutoff class whose limits are Gaussian. The scalar example
does not answer that field-theory question.

\subsection{A flat term changes the exact RG flow}

The nonzero answer still has a Landau pole, at $x=e^{1/g}$, and obeys
$\partial_{\log x}q=q^2$ with its supplied label $g$ held fixed.
Can a function with exactly the same expansion
and reference value have regular running at every scale? This example
allows us to answer directly. To construct a regular response, we
prescribe $\widetilde q(x,g)\to1$ as $x\to\infty$. Choose its reciprocal
to be $1+\log(1+c/x)$: it tends to one at large $x$, while for
$x\ll c$ it approximates $1+\log c-\log x$. This endpoint is
additional information supplied in the completion. The exact reference condition
$\widetilde q(1,g)=g$ fixes $c(g)$, giving, for $0<g<1$,
\begin{equation}
 \widetilde q(x,g)^{-1}
 =1+\log\left(1+\frac{c(g)}x\right),
 \qquad c(g)=e^{1/g-1}-1>0.
 \label{eq:toy-c}
\end{equation}
Since $c(g)>0$, the reciprocal is positive and at least one for every
$x>0$. The completed
function therefore increases from zero in the infrared to one in the
ultraviolet, with no finite-scale pole.

To compare this completion with $q$, set $E_g=e^{1-1/g}$ and rewrite
the same function as
\begin{equation}
 \widetilde q(x,g)
 =\frac{g}{1-g\log x+g\log[1+(x-1)E_g]}.
 \label{eq:toy-completed}
\end{equation}
The added logarithm vanishes at $x=1$, preserving the reference value
exactly.

At each fixed $x>0$, the added logarithm is $O_x(e^{-1/g})$.
Subtracting \eqref{eq:toy-q} from \eqref{eq:toy-completed} then gives
$\widetilde q-q=O_x(g^2e^{-1/g})$, uniformly on fixed compact momentum
intervals. Every weak-coupling coefficient is unchanged. At the original
pole, however, $x=e^{1/g}$ moves with $g$: the added logarithm is then
of order one. It removes the pole without contradicting the
fixed-momentum expansion.

Differentiating \eqref{eq:toy-c} at fixed reference value $g$ gives
the exact beta function of this constructed response:
\begin{equation}
 \frac{\partial\widetilde q}{\partial\log x}
 =\widetilde\beta(\widetilde q),\qquad
 \widetilde\beta(v)=v^2(1-e^{1-1/v}).
 \label{eq:toy-beta}
\end{equation}
It agrees with $v^2$ to every perturbative order, but has an
ultraviolet-attractive zero at $v=1$, with
$\widetilde\beta'(1)=-1$: the predicted value approaches one as the
probing momentum grows. The correction changes the exact flow while
preserving its perturbative beta function. Related flat
modifications of RG flow appear in
\cite{KlaczynskiKreimer2014,MaiezzaVasquez2021}.

\begin{figure}[htbp]
\centering
\begin{tikzpicture}
\begin{axis}[width=.68\linewidth,height=4.1cm,
 xlabel={$v=\widetilde q(x,g)$},ylabel={$\beta(v)=dv/d\log x$},
 xmin=0,xmax=1.15,ymin=-.25,ymax=1.4,
 axis lines=left,legend style={draw=none,font=\small,at={(.03,.97)},anchor=north west},
 samples=130,domain=.002:1.15]
\addplot[black,dashed,thick]{x^2};
\addlegendentry{$v^2$}
\addplot[blue!60!black,thick]{x^2*(1-exp(1-1/x))};
\addlegendentry{$v^2(1-e^{1-1/v})$}
\addplot[only marks,mark=*,blue!60!black] coordinates {(1,0)};
\node[anchor=south,font=\small] at (axis cs:1,.15) {$v=1$};
\end{axis}
\end{tikzpicture}
\caption{A correction smaller than every power near $v=0$ changes the
global energy evolution. The completed flow approaches $v=1$ as
$x\to\infty$. The constructed family attains $0<v<1$, where the
solid completed beta function is positive; its continuation is negative
just above one.
Here $x=P^2/\mu^2$ varies while the supplied reference label
$g$ is fixed; the derivative is with respect to $\log x$, not $\log P$.}
\label{fig:beta-comparison}
\end{figure}
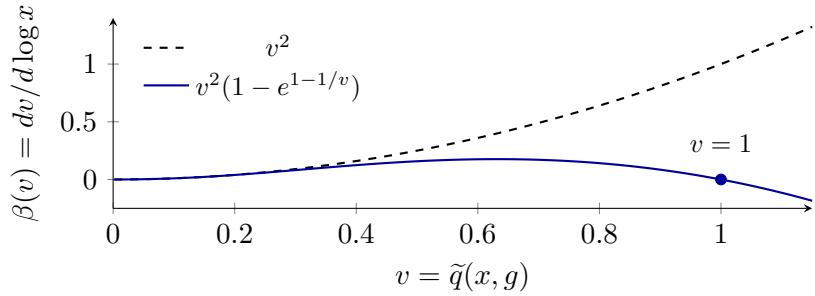

Two successive changes of reference scale should give the same
prediction as using the combined momentum ratio directly. If the
response at $x_2$ is used as the new reference value, this requires
\begin{equation}
 \widetilde q(x_1x_2,g)
 =\widetilde q\bigl(x_1,\widetilde q(x_2,g)\bigr).
 \label{eq:toy-composition}
\end{equation}
Indeed, \eqref{eq:toy-c} implies
$c(\widetilde q(x,g))=c(g)/x$. This proves the composition law:
the successive changes divide $c(g)$ by $x_2$ and then $x_1$.
Together with the exact reference normalization, it shows that the
completion defines a consistent renormalization-group flow.

We have obtained a regular, nonzero function with the same expansion
at each fixed momentum and uniformly on fixed compact momentum
intervals. No nonnegative bare trajectory has this function as its
cutoff limit: every such limit vanishes. The next section obtains these
logarithmic coefficients in a leading-$1/N$ collective two-point
calculation and tests
whether they can be realized by distributions and a generalized free
field satisfying the OS axioms. It also gives \eqref{eq:toy-completed}
a positive spectral interpretation. The finite-$N$ interacting target
additionally requires its nonzero higher connected correlations.

\section{From a perturbative sum to a positive quantum field}
\label{sec:collective}

The second objection is failure of ordinary Borel summation. To test
what that failure leaves open, we use the two-point expansion of a
collective field formed from the singlet composite $\Phi_a\Phi_a$ at
leading order in $1/N$. Bubble diagrams make its coefficients explicit,
so we can compare a summation prescription with the full expansion
and then check the field axioms. This supplies a soluble test using
coefficients from quartic perturbation theory.

At fixed momentum the bubble series has a known geometric sum.
Its Euclidean pole still requires a prescription when the correlation
is integrated against a test function reaching the pole. We derive one
prescription on every Schwartz test, including tests for which the
ordinary Borel integral fails. We then construct two positive spectral
covariances and choose their contact terms to preserve the prescribed
coefficients. Wick's rule gives the higher correlations of generalized
free fields: their connected functions above order two vanish. The
example therefore tests distributional summation and positivity for
the collective expansion. Conjecture~\ref{conj:existence} requires the
distinct massive finite-$N$ elementary expansion, including its
nonzero higher connected functions.

\subsection{Where the expansion comes from}

Use the bare prescription \eqref{eq:bare}, retaining its leading
diagrammatic order in $1/N$. Let $w=P^2$ be the Euclidean squared
channel momentum and let $M$ be the renormalized elementary pole mass
in this calculation. The momentum-independent tadpole contribution
is absorbed into $M$; further momentum-dependent self-energy
corrections are suppressed by $1/N$. With a regulator understood,
the two-propagator bubble and the chain in one four-point channel are
\begin{equation}
 \begin{aligned}
 \Pi_B(w;M)&=\int_{\rm reg}\frac{d^4k}{(2\pi)^4}
       \frac1{(k^2+M^2)((k+P)^2+M^2)},\\
 F_B(w)&=h_B\sum_{r\geq0}\left(-\frac{h_B\Pi_B(w;M)}2\right)^r
       =\frac{h_B}{1+h_B\Pi_B(w;M)/2}.
 \end{aligned}
 \label{eq:bubble}
\end{equation}
The last equality identifies the rational function whose formal
expansion generates the chain; no numerical summation is needed yet.
With $p_1,\ldots,p_4$ denoting the external elementary-field momenta,
the leading-$1/N$ connected amputated four-point amplitude is
\begin{equation}
 \mathcal A_{abcd}
 =-\frac1N\bigl[\delta_{ab}\delta_{cd}F_B((p_1+p_2)^2)
       +\text{the two crossed channels}\bigr]+O(N^{-2}).
 \label{eq:elementary-channel}
\end{equation}
Each successive bubble supplies an index loop compensating the
additional vertex factor $1/N$; its symmetry factor gives the $1/2$
in \eqref{eq:bubble}. The chains therefore retain arbitrarily many
coupling orders within the single displayed order of the $1/N$
expansion. They give explicit elementary four-point channel
coefficients from which to obtain the collective two-point series
\cite{ColemanJackiwPolitzer1974,Romatschke2023}.

The factor $N^{-1/2}$ makes the collective two-point fluctuations
order $N^0$. We include the further factor $h_B/4$ so that their
kernel has the form $h_B/4-F_B/4$, exposing the elementary four-point
channel after a local subtraction. With the vacuum expectation
subtracted, define
\begin{equation}
 \chi_N=\frac{h_B}{4\sqrt N}
       \left(\Phi_{B,a}\Phi_{B,a}
                   -\langle\Phi_{B,a}\Phi_{B,a}\rangle\right)
 \label{eq:chi}
\end{equation}
This remains a formal bare composite at fixed regulator. Diagram
counting gives connected $r$-point coefficients of order
$N^{(2-r)/2}$ for $r\geq2$ at fixed quartic coupling. Retaining the
order-$N^0$ hierarchy therefore leaves only a connected two-point
kernel; its higher correlations use Wick pairings. Its momentum-space
kernel at this retained order is
\begin{equation}
 C_{B,2}=\frac{h_B^2\Pi_B}{8}
           -\frac{h_B^2\Pi_B^2 F_B}{16}
         =\frac{h_B}{4}-\frac{F_B}{4}.
 \label{eq:collective-algebra}
\end{equation}
The first term pairs the elementary fields across the two composite
insertions; the connected elementary four-point channel supplies the
second. The chain identity in \eqref{eq:bubble} combines them into
the final expression.
The momentum constant $h_B/4$ multiplies
$\delta^{(4)}(x-y)$ in position space. Subtracting it is our contact
renormalization convention for this collective example. The remaining
kernel $-F_B/4$ becomes regulator independent under the channel
subtraction below. In the massless case its coefficients use the same
logarithmic powers as the scalar coefficients in
\eqref{eq:toy-renormalized}, at $x=w/\mu^2$, with the overall factor
fixed by the channel normalization below. This fixes the collective expansion and its extension
to coincident insertions before choosing a positive completion. It does not
remove every contact coefficient: $-F_B/4$ itself starts with the
momentum constant $-h_B/4$. After renormalization its order-$g$ term
is therefore a contact, while the separated-point correlation starts
at order $g^2$. The positive covariance will be constructed from
these prescribed coefficients below; the formal calculation has not
yet constructed a continuum composite operator.

For this calculation choose the formal channel subtraction
$F_B(\mu^2)=32\pi^2g$, where $\mu$ is the fixed reference momentum
scale. Expressing $h_B$ as a formal series in $g$ gives
\begin{equation}
 F_M(w)=\frac{32\pi^2g}
 {1+16\pi^2g[\Pi_B(w;M)-\Pi_B(\mu^2;M)]}.
 \label{eq:subtracted-chain}
\end{equation}
The subtraction removes the momentum-independent regulator
divergence in the bubble. Set $M=0$ to isolate the logarithmic
example; the conjecture continues to use a fixed positive elementary
mass. The regulator-independent difference and the prescribed
collective series are
\begin{equation}
 \begin{aligned}
 \Pi_B(w;0)-\Pi_B(\mu^2;0)&=-\frac1{16\pi^2}\log(w/\mu^2),\\
 \widehat W(g;w)&=-A\sum_{j\geq0}[\log(w/\mu^2)]^j g^{j+1},
                         \qquad A=8\pi^2.
 \end{aligned}
 \label{eq:log-series}
\end{equation}
Now choose a numerical $g>0$. It labels the collective example at
fixed $\mu$; varying $w=Q^2$ probes that selected example without
changing $g$. This is not the massive finite-$N$ theory label of
Section~\ref{sec:question}. At each fixed $w>0$, the series converges
for $|g\log(w/\mu^2)|<1$, with sum
\begin{equation}
 -\frac{F_0(w)}4=-\frac{Ag}{1-g\log(w/\mu^2)}
                 =\frac{A}{\log(w/a)},\qquad
 a=\mu^2 e^{1/g}.
 \label{eq:known-kernel}
\end{equation}
The continuation has a positive-Euclidean pole at $w=a$. Its values
away from the pole do not define an integral against every test
function, and the pointwise convergence condition does not hold
throughout an unbounded momentum support. We therefore return to the
coefficient distributions and seek one continuous linear rule for
their smeared sum, including tests crossing the pole.

\subsection{Why examine Borel summation when the geometric sum is known?}

Smearing the coefficients before summing gives ordinary numbers on
which scalar Borel summation can act. The momentum measure is absorbed
into the half-line test: for a radial momentum test
$\varphi(P)=\psi(P^2)$, angular integration gives
$f(w)=\pi^2w\psi(w)$. For
$f\in\Sch([0,\infty))$, smooth up to zero with all derivatives
rapidly decreasing, define the logarithmic moments
\begin{equation}
 c_j[f]=\int_0^\infty[\log(w/\mu^2)]^j f(w)\,dw.
 \label{eq:log-moments}
\end{equation}
Every coefficient integral exists: logarithmic powers are integrable
at zero and Schwartz decay controls infinity. No condition $f(0)=0$
is needed.
Ordinary Borel summation of $-A\sum c_j[f]g^{j+1}$ would use
\begin{equation}
 B_f(\xi):=\mathcal B[\langle\widehat W,f\rangle/g](\xi)
     =-A\sum_{j\geq0}\frac{c_j[f]}{j!}\xi^j
     =-A\int_0^\infty(w/\mu^2)^\xi f(w)\,dw,
 \qquad
 \int_0^\infty e^{-\xi/g}B_f(\xi)\,d\xi.
 \label{eq:borel}
\end{equation}
We transform the series divided by $g$, so the recovery integral has
no $1/g$ prefactor, as in Section~\ref{sec:cutoff}. The Borel series
converges near zero and equals the displayed integral there; that
integral continues it to $\operatorname{Re}\xi>-1$, including the
whole positive ray. Continuation alone does not guarantee that the
recovery integral exists.

For a nonnegative Schwartz $f$ with unbounded support, retain the
part above any $R>a$:
\begin{equation}
 |B_f(\xi)|\geq A(R/\mu^2)^\xi\int_R^\infty f(w)\,dw,
 \qquad \xi\geq0.
 \label{eq:borel-growth}
\end{equation}
Because $\log(R/\mu^2)>1/g$, this positive tail makes
$e^{-\xi/g}|B_f(\xi)|$ grow exponentially. The ordinary Borel
integral diverges for every $g>0$. Gaussian momentum smearing supplies
an explicit example. Angular integration of
$\varphi(P)=e^{-P^2/\mu^2}$ gives $f(w)=\pi^2w e^{-w/\mu^2}$, and
\begin{equation}
 B_f(\xi)=-A\pi^2\mu^4\Gamma(2+\xi).
 \label{eq:gaussian-borel}
\end{equation}
The Gamma function outgrows every fixed exponential. Here ordinary
recovery fails through growth at infinity. In massive scalar
$\phi^4_4$, factorial bounds at fixed external momenta give a locally
analytic Borel transform \cite{DeCalanRivasseau1981}; continuation and
recovery along the whole positive ray require further control.
Renormalons are proposed positive-axis obstructions
\cite{Beneke1999}; the Lipatov bounds of \cite{MagnenEtAl1987} concern
renormalon-free parts and leave exact cancellation unresolved for the
one-component theory. High-loop minimal-subtraction RG data instead
show hints of instanton-dominated growth \cite{DunneMeynig2022}.
These statements do not settle summability of the physically normalized
finite-$N$ elementary hierarchy. Our collective example supplies a
different test: recovery after an explicit growth obstruction.

\subsection{Imaginary rays give a prescription and an error bound}

The positive ray amplifies the momentum tail through
$(w/\mu^2)^\xi$. On an imaginary ray this factor oscillates instead.
Put $t=\log(w/\mu^2)$ and include the change-of-variable Jacobian in
\begin{equation}
 h_f(t)=\mu^2e^t f(\mu^2e^t),\qquad
 \widehat h_f(s)=\int_{\mathbb R}h_f(t)e^{-ist}\,dt.
 \label{eq:fourier-profile}
\end{equation}
The logarithmic moments are now ordinary moments of $h_f$:
$c_j[f]=\int_{\mathbb R}t^j h_f(t)\,dt$. Both $h_f$ and its Fourier
transform are Schwartz. At $t\to-\infty$ the Jacobian controls
the endpoint even if $f(0)\ne0$; at $t\to+\infty$ rapid decay of $f$
controls the other end. The Borel transform on the two imaginary rays
is therefore a rapidly decaying Fourier transform:
\begin{equation}
 B_f(-is)=-A\widehat h_f(s),\qquad
 B_f(is)=-A\widehat h_f(-s).
 \label{eq:imaginary-transform}
\end{equation}
Use this decay to define recovery integrals directly on the rays
$\xi=-is$ and $\xi=is$, each oriented outwards from zero:
\begin{equation}
 \begin{aligned}
 I_\downarrow(f,g)&=-i\int_0^\infty e^{is/g}B_f(-is)\,ds,\\
 I_\uparrow(f,g)&=i\int_0^\infty e^{-is/g}B_f(is)\,ds.
 \end{aligned}
 \label{eq:imaginary-integrals}
\end{equation}
They converge absolutely and depend continuously on every Schwartz
test. No rotation from a convergent positive-ray integral is assumed.
The Fourier identity
$\int_0^\infty e^{isx}ds=\pi\delta(x)+i\PV(1/x)$ evaluates them:
\begin{equation}
 I_\downarrow=\langle W^0,f\rangle+i\pi Aa f(a),\qquad
 I_\uparrow=\langle W^0,f\rangle-i\pi Aa f(a),\qquad
 W^0=A\PV\frac1{\log(w/a)}.
 \label{eq:ray-values}
\end{equation}
The principal value excludes equal shrinking distances on both sides
of $w=a$. The arithmetic mean
$I_{\rm sym}=(I_\downarrow+I_\uparrow)/2$ therefore gives
$\langle W^0,f\rangle$, one real prescription for the already known
kernel on the whole test space.\footnote{These are opposite imaginary
rays, rather than upper and lower detours about one Stokes direction.
Their arithmetic mean differs from the general \textit{\'Ecalle median}
prescription; Appendix~\ref{app:distributional} gives the comparison.}
This use of Borel and Fourier methods is related to
standard distributional summation
\cite{CalicetiGrecchiMaioli1986,CalicetiGrecchiMaioli1987Erratum}.

The prescription must also retain the expansion after smearing.
For every $K\geq1$, integrating \eqref{eq:imaginary-integrals} by
parts $K+1$ times uses
$\widehat h_f^{(j)}(0)=(-i)^j c_j[f]$ and gives, for either ray and
their mean,
\begin{equation}
 \left|I(f,g)+A\sum_{j=0}^{K-1}c_j[f]g^{j+1}\right|
 \leq A g^{K+1}\left(|c_K[f]|+
                   \|\widehat h_f^{(K+1)}\|_{L^1}\right).
 \label{eq:ray-remainder}
\end{equation}
The test-function factor on the right is a continuous Schwartz seminorm,
independent of $g$, so the error is uniform on bounded Schwartz sets.

These half-line statements apply to arbitrary four-dimensional
momentum tests, not just chosen radial examples. For
$\varphi\in\Sch(\mathbb R^4)$ use
\begin{equation}
 (\mathcal R\varphi)(w)
   =\frac w2\int_{S^3}\varphi(\sqrt w\,\omega)\,d\omega,
 \qquad
 \langle W(P^2),\varphi\rangle=\langle W,\mathcal R\varphi\rangle.
 \label{eq:radial}
\end{equation}
Angular averaging is even in $\sqrt w$, so this map is smooth at zero
and continuous into half-line Schwartz space. Fourier transformation
then gives the corresponding position-space distributions.

\begin{theorem}[Borel--Fourier recovery on every Schwartz test]
\label{thm:bf}
Fix $\mu>0$, $A=8\pi^2$ and $g>0$, and put $a=\mu^2e^{1/g}$.
The ray integrals \eqref{eq:imaginary-integrals} converge absolutely
and define continuous linear functionals on $\Sch([0,\infty))$.
Their values are \eqref{eq:ray-values}, so their mean recovers
$W^0=A\PV[\log(w/a)]^{-1}$. For every fixed $K\geq1$, either ray
and their mean satisfy
\[
 \begin{aligned}
 \left|I(f,g)+A\sum_{j=0}^{K-1}c_j[f]g^{j+1}\right|
 &\leq C_{K,\mu}A g^{K+1}q_4(f),\\
 q_4(f)&:=\max_{0\leq j\leq2}\sup_{w\geq0}(1+w)^4|f^{(j)}(w)|,
 \end{aligned}
\]
with $C_{K,\mu}$ independent of $f$ and $g$.
For every $L>0$, the difference $2\pi iAa f(a)$ is $O(e^{-L/g})$
as $g\downarrow0$, uniformly on bounded subsets of the full Schwartz
space; the required decay-weight order may depend on $L$.
The continuous radial map \eqref{eq:radial} and Fourier transformation
transfer these statements to all four-dimensional momentum and
position Schwartz tests.
\end{theorem}
\begin{proof}
Appendix~\ref{app:distributional} proves the endpoint, continuity and
remainder estimates and the ray evaluations.
\end{proof}
Averaging Wick products formed from the two ray values need not give
Wick products of their mean. We therefore choose the mean two-point
kernel first and then form higher correlations by Wick pairings. This prescription at positive real $g$ asserts neither a
unique completion nor summability in a complex angular sector.
Appendix~\ref{app:distributional} also compares it with lateral, median
and accelerated summation.

\subsection{A positive spectrum removes the pole}

The sum $W^0$ acts on all test functions, but its pole prevents the
positive scalar-field reconstruction we seek. A scalar spectral
representation superposes free Euclidean propagators $1/(w+s)$.
Spectral positivity means that their weights are nonnegative, giving
positive state norms. For Minkowski correlations, superpose the
positive-frequency free two-point functions with the same weights.
Their Fourier transforms have support in the forward energy-momentum
cone if $s\geq0$. The Euclidean propagators have no pole
at positive Euclidean $w$, and polynomial contact terms cannot
create or cancel the pole at $a$. A principal value makes smearing
meaningful but does not cure this obstruction.

The logarithmic kernel already has a nonnegative density on its
timelike cut. Keeping that density and subtracting the forbidden pole
gives our first covariance. With the principal logarithm, Cauchy's
formula gives
\begin{equation}
 \begin{aligned}
 \rho_a(s)&=\frac{A}{\log^2(s/a)+\pi^2},\\
 C_D(w)&=\int_0^\infty\frac{\rho_a(s)}{w+s}\,ds
       =\frac{A}{\log(w/a)}+\frac{Aa}{a-w}.
 \end{aligned}
 \label{eq:dispersive}
\end{equation}
Here $s$ is mass squared and $\rho_a(s)\,ds$ is its spectral weight.
The spectral integral converges. The two apparent singularities
cancel at $w=a$, leaving $C_D(a)=A/2$.
The covariance is finite for $w\geq0$, analytic for $w>0$, and bounded
between zero and $A$. This is the analytic-perturbation-theory kernel
\cite{ShirkovSolovtsov1997,MiltonSolovtsov1997}; removing unphysical
propagator poles also has earlier precedents
\cite{Redmond1958,GhinculovEtAl1997}.

The covariance $C_D$ has a weak-coupling expansion starting with a
new order-zero contact coefficient. At fixed momentum the added
$Aa/(a-w)$ tends to the constant $A$, which
\eqref{eq:log-series} does not contain. To compare with that prescribed
expansion on all Schwartz tests, choose the Euclidean extension
\begin{equation}
 W_D=C_D-A,\qquad
 W^0-W_D=A\PV\frac{w}{w-a}.
 \label{eq:contact-match}
\end{equation}
The positive covariance $C_D$ constructs the field. The subtraction
$C_D\mapsto W_D$ is a contact change of size $A$, independent of
$g$; it is not flat. The difference we will prove flat is
$W_D-W^0$. In particular, $W_D$ is not the covariance of a positive
Gaussian probability measure.
It gives the same quantum field as $C_D$ on the standard ordered OS
test domain, where each group has support in $0<t_1<\cdots<t_r$:
no two insertion times in either ordered group coincide,
and a reflected negative time cannot coincide with a positive time.
An OS reflection form tests correlations between positive-time probes
and the time reflection of other positive-time probes.
Every contact contraction consequently vanishes in the reflection
forms, including those between different field orders. Thus $W_D$
specifies a different extension at coincidence of the same
reconstructed quantum field. Perturbative matching instead tests
this extension on arbitrary Schwartz functions, including ones that
reach coincidences. Matching fixes the contact coefficients order by
order and still permits additional flat contact terms, as in
Section~\ref{sec:question}. The extension convention must remain
explicit when its Fourier transform is used to define a response;
Section~\ref{sec:poles} returns to this issue.

To check that no perturbative coefficient has changed, a fixed-momentum
estimate is insufficient near the moving pole. In fixed momentum units,
let $p_3(f)=\max_{j=0,1}\sup_{w\geq0}(1+w)^3|f^{(j)}(w)|$.
For fixed $\mu$, the difference in \eqref{eq:contact-match} obeys
\begin{equation}
 |\langle W^0-W_D,f\rangle|
 \leq C_\mu A e^{-1/g}p_3(f).
 \label{eq:dispersive-flat}
\end{equation}
Here and below constants are independent of the probe and sufficiently
small $g$. One can see the bound by splitting the integral below
$a/2$, across $[a/2,3a/2]$, and above $3a/2$. Below the pole the
denominator supplies $1/a$. Across it, subtract $af(a)$ from $wf(w)$
inside the principal value; the first derivative and Schwartz decay
again give $O(p_3(f)/a)$. Above it, the tail is $O(p_3(f)/a^2)$.
Since $1/a=\mu^{-2}e^{-1/g}$, all three contributions are flat.

The nonnegative spectral measure in \eqref{eq:dispersive}, supported
on $s\geq0$, together with Wick pairings constructs a generalized free field
\cite{Greenberg1961,StreaterWightman2000}. For the full-Schwartz
comparison, define $S^D_{2r}$ by Wick pairings of $W_D=C_D-A$, with
odd correlations zero. On the OS test domain this hierarchy agrees
with the one constructed from $C_D$, so it reconstructs the same
positive quantum field. The bounded covariance gives the usual
factorial growth bounds. Smeared two-point functions cluster, and Wick
pairings give clustering for the hierarchy. The hierarchy with this
contact convention also preserves every prescribed coefficient of the
massless leading-$1/N$ collective Wick hierarchy. Each Wick pairing assigns every
insertion to one pair, so its contribution is a tensor product of
two-point distributions in separate pair variables. At each fixed
Wick order, subtracting the two products gives a finite sum with one
flat two-point difference in each term. The flatness bound for that factor
and continuity bounds for the remaining factors give the estimate on
arbitrary joint Schwartz tests, including tests that do not factorize.
We first select $W^0$ as the two-point
sum, choose the positive covariance $C_D$ and its matching extension
$W_D$, and then form their respective Wick hierarchies. We do not
average the higher hierarchies obtained from the
individual ray values in \eqref{eq:ray-values}.
Appendix~\ref{app:gaussian} collects the joint-test, contact-support,
growth and clustering arguments for the whole Wick hierarchy.

The reference normalization changes by a flat term in this completion:
\begin{equation}
 W_D(\mu^2)=-Ag+\frac{A}{e^{1/g}-1}.
 \label{eq:dispersive-reference}
\end{equation}
Here $g$ is the label inherited from the formal channel subtraction.
The exact two-point value $-W_D(\mu^2)/A$ differs from it, although
their expansions agree. One can instead use that exact value as the
parameter: this flat reparametrization preserves all coefficients.
Explicitly, $g_D=g-(e^{1/g}-1)^{-1}$ is locally invertible near zero,
since its derivative tends to one. Comparing the two completions at
the same exact subtraction value therefore remains possible.
Section~\ref{sec:poles} examines its running. This example establishes coefficient matching; it does not impose
the massive elementary-field normalization of the conjecture.

\subsection{The repaired RG flow supplies a second positive completion}

The flow of Section~\ref{sec:cutoff} offers a different choice with the
same coefficients. Since $\widetilde q\to1$ in the ultraviolet,
$A(1-\widetilde q)$ decays there and is a candidate covariance.
For $0<g<1$, set
\begin{equation}
 b=\frac ae-\mu^2>0,\qquad
 C_R(w)=A[1-\widetilde q(w/\mu^2,g)]
       =A\left[1-\frac1{1+\log(1+b/w)}\right].
 \label{eq:rg-covariance}
\end{equation}
Equivalently, $b=\mu^2c(g)$ with $c(g)$ from \eqref{eq:toy-c}, which
connects this spectral scale directly to the repaired flow.
Pointwise positivity alone would not establish reflection positivity.
The explicit spectral representation does:
\begin{equation}
 \begin{aligned}
 C_R(w)&=\frac{Z}{w+s_0}
          +\int_0^b\frac{\rho_R(s)}{w+s}\,ds,\\
 \rho_R(s)&=\frac{A}{[1+\log(b/s-1)]^2+\pi^2},\quad 0<s<b,\\
 s_0&=\frac{eb}{e-1},\qquad Z=\frac{Aeb}{(e-1)^2}>0.
 \end{aligned}
 \label{eq:rg-spectrum}
\end{equation}
The logarithm has a cut on $[-b,0]$ and one denominator zero at
$-s_0$. Its cut discontinuity gives the nonnegative density above,
and the residue gives the positive atom $Z$. The covariance decays
as $Ab/w$ at infinity, so a Cauchy contour supplies no further term.

Again define the matching extension by $W_R=C_R-A$. It has the same
complete smeared expansion as $W^0$:
\begin{equation}
 |\langle W_R-W^0,f\rangle|
 \leq C_\mu A e^{-1/g}p_3(f),\qquad W_R(\mu^2)=-Ag.
 \label{eq:rg-flat}
\end{equation}
For the estimate, below $a/2$ write the two kernels as
\[
 \begin{aligned}
 W_R(w)&=-\frac{A}{D_R(w)},&
 D_R(w)&:=1+\log(1+b/w),\\
 W^0(w)&=-\frac{A}{D_0(w)},&
 D_0(w)&:=\log(a/w).
 \end{aligned}
\]
Here $D_R\geq1$ and $D_0\geq\log2$, while
\[
 D_R(w)-D_0(w)=\log[1+e(w-\mu^2)/a]
                =O((w+\mu^2)/a).
\]
The denominators are bounded away from zero, so this also controls
the difference of their reciprocals on the lower interval.
The same principal-value subtraction and tail split used for
\eqref{eq:dispersive-flat} control the other regions. The hierarchy
$S^R_{2r}$ defined by Wick pairings of $W_R$ therefore has the same
collective coefficients. On the OS test domain it agrees with the
positive hierarchy constructed from $C_R$, giving a second generalized
free field in the stated contact convention.

This second choice has a direct link to the repaired scalar flow:
its normalized two-point value is exactly
$-W_R(w)/A=\widetilde q(w/\mu^2,g)$, including at the reference scale.
At $w=Q^2$ this is a prediction as the probing energy varies, while
the supplied reference label $g$ stays fixed. The selected field keeps
its spectral parameter $b$ fixed; the dispersive field instead keeps
$a$ fixed. Neither value is an elementary quartic coupling: each
constructed collective field has vanishing connected four-point
function.
The two-point spectral measures differ as well: $C_R$ has a finite
continuous mass-squared interval and an isolated mass, while $C_D$
has continuous mass weight extending to infinity and decays as
$A/\log(w/a)$. At equal inherited $g$, their exact reference values
differ. For a common small exact value $\eta$, let $\gamma$ be the local
inverse of $g_D(g)=g-(e^{1/g}-1)^{-1}$ and use inherited labels
$\gamma(\eta)$ for $D$ and $\eta$ for $R$. Then both satisfy
$W_j(\mu^2)=-A\eta$, and the flat change
$\gamma(\eta)-\eta=O(e^{-1/\eta})$ preserves all coefficients.
Their spectral scales must be evaluated at their respective labels;
Appendix~\ref{app:positive} gives this exact reference comparison.

Table~\ref{tab:kernels} separates the roles of the kernels. Appendix~\ref{app:positive}
gives the spectral and flatness checks for both completions, including
the isolated positive mass in $C_R$.
\begin{table}[htbp]
\centering
\small
\caption{Summation, positive covariances and the contact convention used for matching.}
\label{tab:kernels}
\begin{tabular}{@{}>{\raggedright\arraybackslash}p{.22\linewidth}>{\raggedright\arraybackslash}p{.33\linewidth}>{\raggedright\arraybackslash}p{.39\linewidth}@{}}
\toprule
Kernel & Role & Relation to the expansion and quantum field\\
\midrule
$W^0$ & Distribution recovered by the ray mean &
Has the logarithmic expansion, but its Euclidean pole prevents a positive scalar spectrum.\\[3pt]
$C_D,\ C_R$ & Positive spectral covariances &
Construct two different generalized free fields; their weak-coupling expansions start with the contact coefficient $A$.\\[3pt]
$W_D=C_D-A$, $W_R=C_R-A$ & Euclidean extensions used for coefficient matching &
Preserve every collective Wick coefficient and reconstruct the same field as the respective $C_D$ or $C_R$.\\
\bottomrule
\end{tabular}
\end{table}

\begin{proposition}[Two positive completions with the collective expansion]
\label{prop:completions}
Fix $\mu>0$, $A=8\pi^2$ and $0<g<1$, with
$a=\mu^2e^{1/g}$ and $b=a/e-\mu^2$.
The covariances $C_D,C_R$ of \eqref{eq:dispersive} and
\eqref{eq:rg-covariance} have nonnegative spectral measures supported
on $s\geq0$, no positive-Euclidean pole, and Wick hierarchies
satisfying the OS axioms, including reconstruction growth and clustering.
The hierarchies of $W_j=C_j-A$, $j=D,R$, give the same respective
reflection forms and reconstructed fields on the stated strictly
ordered OS domain excluding coincidences. The subtraction is a
$g$-independent contact term and is not flat.

For all sufficiently small $g$,
\[
 |\langle W_j-W^0,f\rangle|
 \leq C_\mu A e^{-1/g}p_3(f),\qquad f\in\Sch([0,\infty)),
\]
with $C_\mu$ independent of $g$ and $f$.

After radial transfer and Fourier transformation, at every fixed
Wick order $2r$ the hierarchies formed from $W_j$ and $W^0$ differ by
$O(e^{-1/g})$, uniformly on bounded sets of arbitrary joint Schwartz
tests. Thus they preserve every massless leading-$1/N$ collective
Wick coefficient. The distinct spectra remain distinct after the
flat relabelling to one exact reference value described above.
\end{proposition}
\begin{proof}
Appendix~\ref{app:positive} gives the spectral representations,
flatness estimates and exact reference matching;
Appendix~\ref{app:gaussian} gives the joint-test and reconstruction checks.
\end{proof}
The positive completion supplies information beyond the imaginary ray
choice, which changes only the term sampling $f(a)$.
The interacting finite-$N$ elementary hierarchy remains to be constructed.

\section{Which poles can a quantum theory have?}
\label{sec:poles}

The preceding construction removes the pole from a two-point function.
What does that establish about the third objection, the Landau pole
of quartic theory? The answer depends on what has a pole. A parameter
used to label a regular theory can diverge. A timelike correlation
pole can describe a particle. For a Euclidean four-point subtraction
response, finiteness at each probing scale is a further property to
establish. We examine these three cases to identify the regularity conditions
needed for the four-point coupling defined below.

\subsection{A singular label and a regular value of the correlation}

Keep the dispersive field of \eqref{eq:dispersive} fixed, including its
spectral scale $a$. At the fixed reference momentum $\mu$, its inherited
label is $g=1/\log(a/\mu^2)$. At another probing scale $Q$, the same
field has the inherited running coordinate
\begin{equation}
 \bar g_{\rm inh}(Q;a)=\frac1{\log(a/Q^2)},\qquad
 \frac{d\bar g_{\rm inh}}{d\log Q^2}=\bar g_{\rm inh}^{\,2}.
 \label{eq:inherited-running}
\end{equation}
The supplied reference label $g$ stays fixed while $Q$ varies.
The inherited coordinate has
a pole at $Q^2=a$, whereas the covariance does not:
$C_D(a)=A/2$. The reciprocal coordinate
$1/\bar g_{\rm inh}=\log(a/Q^2)$ passes regularly through that scale.
Thus the coordinate inherited from the formal channel sum diverges
while the selected field stays regular. For $Q^2\ne a$ it encodes
the same fixed spectral scale through
$a=Q^2\exp(1/\bar g_{\rm inh})$. Above $Q^2=a$ this coordinate is
negative; the covariance and its positive spectrum are unchanged.

In the contact convention of Section~\ref{sec:collective}, use the
completed correlation itself to define the predicted value:
\begin{equation}
 v(Q;a)=1-\frac{C_D(Q^2)}A=-\frac{W_D(Q^2)}A
        =\bar g_{\rm inh}(Q;a)-\frac1{e^{1/\bar g_{\rm inh}(Q;a)}-1}.
 \label{eq:exact-two-point-value}
\end{equation}
The last expression holds for $Q^2\ne a$; its continuous extension
at $Q^2=a$ is $v=1/2$. The reference values $g$ and
$g_D=v(\mu;a)$ stay fixed. For a small positive inherited value the
change is flat. Near the apparent pole both terms on the right
diverge, but their difference is regular. In $t=\log(Q^2/a)$,
\begin{equation}
 v=V(t)=\frac1{1-e^{-t}}-\frac1t,\qquad
 V(t)=\frac12+\frac{t}{12}+O(t^3).
 \label{eq:regular-running}
\end{equation}
Moreover,
\begin{equation}
 V'(t)=\frac1{t^2}-\frac1{4\sinh^2(t/2)}>0\quad(t\ne0),\qquad
 V'(0)=\frac1{12}.
 \label{eq:regular-running-derivative}
\end{equation}
The normalized value runs smoothly from zero to one, with $v=1/2$
at the old pole. Since $V$ is strictly increasing, its beta function
can be written
$\beta_v(v)=V'(V^{-1}(v))$ for $0<v<1$, using $\log Q^2$ as the
scale variable. It has the perturbative expansion $v^2$ near zero,
but no finite-scale pole. It approaches the ultraviolet endpoint as
$1-v\sim1/\log(Q^2/a)$. The RG-derived completion of
Section~\ref{sec:cutoff} approaches its endpoint faster, with
$\widetilde\beta'(1)=-1$. Their common perturbative coefficients
determine neither exact flow.
Appendix~\ref{app:massive} considers a massive version of the
bubble-chain example and another way to obtain a positive kernel
without a Euclidean pole.

This is a two-point normalization of a generalized free collective
field, whose connected four-point function vanishes. Its lesson for
the conjecture is a procedure: complete the correlation, impose the
exact subtraction condition on that correlation, and then follow its
value as the energy changes in the selected theory. This exact
normalized value remains regular where the inherited coordinate diverges.

\subsection{A timelike pole can be physical}

A positive spectral atom contributes $R/(w+m_{\mathrm{bs}}^2)$ to a
Euclidean covariance, with $R,m_{\mathrm{bs}}^2>0$. This is regular
for $w>0$, as the atom in \eqref{eq:rg-spectrum} illustrates. At
timelike squared momentum $\sigma$, put $w=-\sigma-i0$ to obtain
\begin{equation}
 \frac{R}{m_{\mathrm{bs}}^2-\sigma-i0}
 =R\PV\frac1{m_{\mathrm{bs}}^2-\sigma}
       +i\pi R\delta(\sigma-m_{\mathrm{bs}}^2).
 \label{eq:physical-timelike-pole}
\end{equation}
It has a finite pairing with each fixed Schwartz or compactly
supported smooth packet in $\sigma$. The principal value subtracts
the packet's value at the pole, and the delta term evaluates that
value. The pole represents a particle or bound state in a
well-defined quantum theory. Such bound-state poles occur in
constructed two-dimensional scalar models, for example the
$\lambda({:\!\phi^6\!:}-{:\!\phi^4\!:})$ theory at sufficiently small
positive $\lambda$ \cite{DimockEckmann1976,Koch1979}.
No bound uniform as the packet width tends to zero is required.

A coupling defined from a timelike response can inherit this particle
pole. Appendix~\ref{app:timelike} gives an explicit composite-response
example in a constructed interacting theory. This less usual
convention defines a coupling from a timelike composite response;
it is a mathematical curiosity compared with the elementary
Euclidean four-point subtraction used for the conjecture. Its
particle pole illustrates an allowed coupling singularity without
answering the Euclidean Landau-pole objection.

\subsection{The physical subtraction coupling and its regularity}

For a theory $\mathcal T_g$, define a four-point subtraction coupling
from its scalar Schwinger functions, using the projection, normalization
and contact convention fixed in Section~\ref{sec:question}. Let
$H_g=\mathcal P G_{4,c}$ be the fixed component projection of the
connected four-point Fourier distribution and $D_{2,g}$ the exact
scalar two-point Fourier distribution, with their overall conservation
deltas removed. Let $p_i^*$ denote the fixed symmetric Euclidean
subtraction momenta at the reference scale $\mu$. Use their scaled values
$p_i(Q)=(Q/\mu)p_i^*$ and independent variables
$\xi=(p_1,p_2,p_3)\in\mathbb R^{12}$, with $p_4=-p_1-p_2-p_3$.
Write $\xi(Q)$ for their scaled configuration. The packet averages below smear these Euclidean distributions before
we divide by the four external propagation factors.

Take any fixed smooth compactly supported packet shape $\varphi$
with integral one; such packets are Schwartz functions. A packet
centred at $x$ with shrinking width is
\begin{equation}
 \varphi_{\epsilon,x}(y)=\epsilon^{-d}
           \varphi\bigl((y-x)/\epsilon\bigr),\qquad \epsilon>0.
 \label{eq:subtraction-packet}
\end{equation}
Use $d=12$ for a numerator shape $\psi$ and $d=4$ for each external
shape $\eta_i$. Thus the numerator is tested in all independent
momentum components. Require each average to have a finite limit
independent of the allowed packet shape, and require every external
limiting value to be nonzero. Then define
\begin{equation}
 \bar g(Q;g)=\mathcal N
 \lim_{\epsilon,\epsilon_i\downarrow0}
 \frac{\langle H_g,\psi_{\epsilon,\xi(Q)}\rangle}
      {\prod_{i=1}^4\langle D_{2,g},
                       \eta_{i,\epsilon_i,p_i(Q)}\rangle}.
 \label{eq:exact-four-point-coupling}
\end{equation}
Here $\mathcal N$ is the fixed numerical normalization; widths may
shrink independently. The nonzero external limits make the ratios
defined once the widths are sufficiently small. Smearing precedes division, so no quotient or
restriction of singular distributions is assumed. Requiring all
packet shapes prevents a symmetric packet from concealing a pole.
The exact reference condition is $\bar g(\mu;g)=g$. That supplied
label stays fixed while $Q$ probes the selected theory; existence
of its limits at other scales is a separate check.

If the numerator distribution is locally given by a continuous function
near the probing curve and the external distributions are locally given
by continuous nonzero functions, the quotient is continuous and bounded
on every compact interval of $Q>0$. These sufficient conditions exclude a finite Euclidean
pole. Analyticity gives stronger regularity, but is unnecessary for
this boundedness. Existence of a point value at one scale alone does
not control nearby scales.

Ruelle's result gives sufficient conditions for the connected
four-point function to be analytic. To apply it here, additionally
require the Euclidean distributions to agree near the chosen
nonexceptional configurations, including the same contact terms, with
the continuation of those fields' time-ordered correlations. If his
regularization and compatible time-ordering hypotheses also hold,
the connected numerator is analytic in all independent momenta
\cite[Sections~4--5]{Ruelle1961}. With analytic nonzero external
propagators, the quotient \eqref{eq:exact-four-point-coupling} is then
analytic at every finite $Q>0$: a pole in its perturbative sum must be
absent from that exact response. Appendix~\ref{app:gaussian} states the
hypotheses and argument. The gap in a general time-ordering
construction is noted in \cite[Section~10]{KravchukQiaoRychkov2021}.
Contacts invisible to reconstruction may change an off-shell convention,
but cannot alter a fixed physical measurement of the reconstructed
theory; the examples in Appendix~\ref{app:gaussian} concern this
convention freedom.

If the packet conditions fail, this prescription need not give a
finite point coupling. Tempered correlations still give finite answers
on each fixed Schwartz packet. If an external propagator has zero
point value, extending the amputated limit requires a separate
cancellation check. Nor is every failed point limit a Landau pole.

Our definition of $\phi^4$ in Section~\ref{sec:question} therefore does
not by itself exclude a pole in the point-valued physical four-point
coupling: OS reconstruction and quartic coefficient matching do not
supply the additional conditions above. We expect regular responses
at nonexceptional Euclidean momenta, and the Hilbert-space and
time-ordering properties used in Ruelle's argument, for the massive
quartic theory envisaged here. A genuine pole would require a failure
of the sufficient regularity conditions above, an exotic possibility
that would need its own explanation.

Our main response to the Landau-pole objection is the possibility that
non-perturbative corrections remove the pole predicted by the
perturbative sum. The flat corrections in Sections~\ref{sec:cutoff}
and~\ref{sec:collective} preserve the perturbative coefficients while
removing the pole from the function example and the collective
covariance. These examples illustrate the repair mechanism; they do
not establish it for the interacting elementary four-point function.
The perturbative pole alone does not rule out such a repair. The next
section asks how a larger theory could supply the interacting scalar
hierarchy.

\section{A conditional route to interacting scalar correlations}
\label{sec:parent}

The collective examples have no higher connected correlations.
To reach the conjectured interacting theory, we need a mechanism that
supplies those correlations together with positivity. One possibility
is to retain the scalar observables of a larger, already constructed
theory, which we call the \emph{parent}. Their correlations would
inherit its OS axioms as a whole hierarchy. Under the estimates below,
its higher connected scalar correlations would be nonzero and its
expansion would be precisely the quartic one.

We describe a heavy-particle route to that agreement. It requires an
actual parent family, agreement of its scalar perturbative coefficients
with the quartic coefficients as the heavy masses grow, and control of
the exact error after truncating the parent expansion. A conditional
theorem states sufficient estimates. The selected mass path makes
inverse-mass corrections flat while placing the heavy thresholds below
the formal scalar Landau scale. The gauge--Higgs--fermion candidate
below checks the compatibility of a one-loop ultraviolet trajectory
and a stable massive tree vacuum for $N=1$; the parent construction,
corrected trajectory and both estimates remain to be established.

\subsection{Retain a physical scalar and vary the supplied parameters}

Suppose the parent contains a local physical scalar multiplet $\Phi_a$
with exact $O(N)$ symmetry. Retain every correlation containing only
these fields. Reflection positivity follows because scalar observables
are a subset of the parent's observables; covariance, clustering and
the reconstruction growth bounds restrict in the same way. No
scalar-only action is needed, and no correlation is changed.

OS reconstruction gives the Hilbert space generated by scalar fields
acting on the vacuum; this space need not contain every parent state.
Heavy channels accessible to those scalar insertions can nevertheless
appear in its spectrum. At each fixed positive $g$, their finite
energy scales are compatible with the conventional effective-field-theory
reading: additional particles affect scalar responses above heavy
thresholds. The restricted correlations still define an exact continuum
hierarchy generated by the retained scalar field (or multiplet).
Under both estimates below, this hierarchy meets the same stated
axiomatic existence criterion of Section~\ref{sec:question}: the OS
axioms, exact physical normalization and every quartic coefficient.
That criterion specifies neither the complete spectrum nor a pure-scalar
microscopic action. A continuum limit of a specified pure-scalar lattice
measure is a different construction requirement.

In a gauge theory the word \emph{physical} matters. A gauge-fixed
charged Higgs propagator need not satisfy the positivity required for
this restriction. Our candidate retains a real gauge-singlet scalar;
its Higgs doublet supplies masses to the additional gauge particles.

In the $N=1$ candidate of Section~\ref{subsec:caf-parent}, the retained scalar $\phi$
couples to two heavy $SU(2)$ fundamental Dirac fermions $Q_\pm$.
Their Minkowski Yukawa interaction is
\[
 -iy\phi(\bar Q_+\gamma_5Q_+-\bar Q_-\gamma_5Q_-).
\]
The simultaneous transformation $\phi\mapsto-\phi$,
$Q_+\leftrightarrow Q_-$ preserves this interaction and, with equal
fermion masses, realizes the required $O(1)$ sign symmetry. These
gauge-charged fermions connect the scalar's quartic running to the
gauge-assisted flow through its Yukawa interaction.

Fix the subtraction scale $Q_0=\mu$ once and for all. The physical
quartic $g$ and the additional renormalized couplings $u_i$ supplied
there label a parent theory. Select a weak-coupling family by writing
\begin{equation}
 u_i(g)=g z_i(g),\qquad z_i(g)=\frac{u_i(g)}g .
 \label{eq:supplied-ratios}
\end{equation}
Thus $z_i$ are ratios of supplied couplings. Gauge and Yukawa
interactions are counted through their squares, so this notation
includes $e^2,y^2=O(g)$. Varying $g$ at fixed $Q_0$ selects different
theories. Each selected theory keeps its input labels fixed while
its energy-dependent coupling predictions vary with the probing
energy $Q$. Schematically, in a chosen renormalization scheme,
\begin{equation}
 \frac{d c(Q;g)}{d\log Q}=\beta(c(Q;g)),\qquad g\text{ held fixed}.
 \label{eq:parent-energy-running}
\end{equation}
Here $c$ is the vector of running predictions in that scheme. It is
distinct from the supplied ratios $z(g)$. In a physical subtraction
scheme the beta function can also depend on energy-to-mass ratios;
this dependence and matching across heavy-particle scales are
suppressed in the schematic equation.

We seek parents with \emph{complete asymptotic freedom} (CAF), also
called total asymptotic freedom: all their dimensionless interaction
couplings tend to zero as $Q\to\infty$ in each selected theory.
Section~\ref{subsec:caf-parent} explains why this property motivates a construction
programme. The supplied ratios must select a CAF family that the parent theory
actually possesses; not every choice of $z$ is allowed. Fixed ratios
would suffice if they select such a family. Otherwise $z(g)$ includes
the matching and scheme-conversion corrections needed to follow such
CAF trajectories with the prescribed physical normalization at $Q_0$.
It is not a freely tunable function for engineering coefficients.
We require these ratios to stay bounded in a fixed compact set,
without requiring every point of that set to be CAF. The conditional
implication below uses the parent's OS axioms, exact normalization
and two estimates; CAF is its motivation, not a separate premise
needed for that implication.

\subsection{Two estimates that would give the quartic expansion}

Let $T_n(g,z,M_H)$ be the parent's exact scalar correlations, with
$n$ field insertions and component indices suppressed. Their light
scalar pole mass $m>0$, unit pole residue and physical quartic value
$g$ satisfy exactly the conditions of Section~\ref{sec:question}.
All additional physical particles have masses comparable to a common
heavy scale $M_H$. Their positive mass ratios stay in a fixed compact
set along the chosen family; these prescribed ratios are suppressed
in the notation. $M_H$ is unrelated to the constituent mass $M$ in
Section~\ref{sec:collective}.

Choose a consistent normalized perturbative prescription for this
parent. At fixed ratios $z$ and fixed consistent mass data, denote
its formal scalar expansion by
\begin{equation}
 \widehat T_n(g;z,M_H)=\sum_{k=0}^{\infty}g^k b_{n,k}(z,M_H).
 \label{eq:parent-formal-series}
\end{equation}
Here $k$ is the perturbative order. Where a joint formal coupling
prescription is available with those mass data held fixed, the
finite total-degree grouping is explicitly
\begin{equation}
 \begin{aligned}
 \widehat T_n^{\rm joint}(g,u;M_H)
   &=\sum_{j\geq0}\sum_{\alpha\in\mathbb N_0^s}
       g^j u^\alpha B_{n;j,\alpha}(M_H),\\
 b_{n,k}(z,M_H)
   &=\sum_{\substack{j\geq0,\ \alpha\in\mathbb N_0^s\\j+|\alpha|=k}}
       z^\alpha B_{n;j,\alpha}(M_H),\qquad u_i=gz_i .
 \end{aligned}
 \label{eq:parent-total-degree}
\end{equation}
Here $j$ counts powers of the quartic coupling, while the multi-index
$\alpha$ lists the powers of the $s$ additional coupling variables.
Write $|\alpha|=\sum_i\alpha_i$ and
$u^\alpha=\prod_i u_i^{\alpha_i}$; the $B_{n;j,\alpha}$ are the
joint coefficient distributions. At each fixed $k$, this substitution
only collects finitely many formal terms; no infinite series in the
other couplings is being summed. In a broken Higgs parameterization,
masses and couplings
can obey relations that must be imposed consistently when defining
\eqref{eq:parent-formal-series}. The schematic joint grouping applies
where that prescription exists; it asserts no polynomial dependence
on arbitrary independently varied physical masses and ratios.

The chosen family subsequently evaluates $b_{n,k}$ at
$z(g),M_H(g)$. Boundedness of a general $z(g)$ does not give it a
Taylor or asymptotic expansion. The resulting terms therefore need
not themselves form an ordinary $g$-series with constant
coefficients. The theorem obtains the fixed quartic coefficients
$a_{n,k}$ directly from estimates on these evaluated coefficients
and the exact remainder, without expanding $z(g)$.

Choose a fixed reference energy $m_*>0$ and the mass path
\begin{equation}
 M_H(g)=m_*\exp[(\log(1/g))^2],\qquad
 \ell_H(g)=1+\log\frac{M_H(g)}{m_*},\qquad 0<g<1 .
 \label{eq:heavy-path}
\end{equation}
An inverse power of $M_H(g)$ times any fixed power of its logarithm
then becomes flat in $g$. At the same time,
$g\ell_H(g)^q\to0$ for every fixed finite $q\geq0$, which controls
logarithmic losses in the remainder. Every positive $g$ selects a
parent with a finite heavy mass; the limit changes theories rather
than taking an infinite-mass limit in one theory at fixed $g$.

The two estimates have different meanings. The first says that the
parent and quartic coefficient distributions become close after
matching the same physical scalar conditions. The second bounds the
actual error when the parent's exact correlation is approximated by
a finite sum of its own coefficients. Agreement of coefficients
alone cannot control this exact error. Both estimates must hold
along the same path, including its selected ratios and masses.

To express their control on tests, in fixed coordinate units use the
standard Schwartz seminorms
\[
 p_{a,b}(f)=\max_{|\gamma|\leq b}\sup_{X\in\mathbb R^{4n}}
             (1+|X|)^a|\partial^\gamma f(X)|,
 \qquad a,b\in\mathbb N_0.
\]
Here $X=(x_1,\ldots,x_n)$ collects the insertion coordinates and
$\gamma$ is a derivative multi-index in $4n$ dimensions. The indices
$a,b$ select finite decay-weight and derivative orders.
For an estimate at correlation order $n$ and perturbative or
truncation index $j$, write $p_{n,j}=p_{A_{n,j},B_{n,j}}$ for a
suitable choice of those orders. Thus $n,j$ select the estimate,
not the standard weight and derivative orders. Existence of some
finite $A_{n,j},B_{n,j}$ suffices; increasing them if necessary gives
one seminorm and one constant dominating both estimates at that
index. The choices and constants must be independent of $g$ along
the selected path. No bound is requested for independently varied
parent inputs. Each logarithmic exponent below is a fixed finite
nonnegative number for its particular correlation and perturbative
or truncation order, independent of $g$ and the varying masses.
No common bound over all orders is required.

\begin{theorem}[Conditional scalar construction]
\label{thm:parent}
Assume an actual parent family contains
$(g,z(g),M_H(g))$ for every sufficiently small $g>0$, has the
physical scalar and exact normalization just specified, and satisfies
the OS axioms, including the reconstruction growth condition.
Use the same prescribed coincident-point extensions in the parent
coefficient comparison and in quartic perturbation theory.
Assume the following estimates for every Schwartz test $f$ and
sufficiently small $g$ along this selected family:
\begin{enumerate}
\item For every fixed $n,k$, the coefficient difference satisfies
\begin{equation}
 |\langle b_{n,k}(z(g),M_H(g))-a_{n,k},f\rangle|
 \leq C_{n,k}p_{n,k}(f)
       \left(\frac{m_*}{M_H(g)}\right)^{\delta_{n,k}}
       \ell_H(g)^{d_{n,k}},\qquad \delta_{n,k}>0 .
 \label{eq:coefficient-decoupling}
\end{equation}
\item For every fixed $n,K$ with $K\geq1$, the exact remainder satisfies
\begin{equation}
 \begin{aligned}
 &\left|\left\langle T_n(g,z(g),M_H(g))
       -\sum_{k=0}^{K-1}g^k b_{n,k}(z(g),M_H(g)),f\right\rangle\right|\\
 &\hspace{2em}\leq C_{n,K}p_{n,K}(f)g^K\ell_H(g)^{r_{n,K}} .
 \end{aligned}
 \label{eq:parent-remainder}
\end{equation}
\end{enumerate}
The constants and seminorm orders are independent of $g$ along the
path, and $d_{n,k},r_{n,K}$ are fixed finite nonnegative exponents.
Then the unchanged scalar hierarchy
\begin{equation}
 S_n(g)=T_n(g,z(g),M_H(g))
 \label{eq:scalar-restriction}
\end{equation}
satisfies Conjecture~\ref{conj:existence}: it obeys the OS axioms and
exact normalization, has the complete quartic expansion
\eqref{eq:matching}, and is interacting for sufficiently small $g>0$.
\end{theorem}

\begin{proof}
Fix $n,K$ and a Schwartz test $f$, and put $t=\log(1/g)$.
The coefficient bound is a fixed test seminorm times
\begin{equation}
 e^{-\delta t^2}(1+t^2)^d=o(g^J)
 \quad\text{for every fixed }J\geq0.
 \label{eq:flat-heavy-error}
\end{equation}
For this proof abbreviate
$T_n^g=T_n(g,z(g),M_H(g))$ and
$b_{n,k}^g=b_{n,k}(z(g),M_H(g))$. We retain one extra perturbative
term so that the extra power of $g$ absorbs the remainder's logarithmic
loss. Split the target error exactly as
\[
 \begin{aligned}
 T_n^g-\sum_{k=0}^{K-1}g^k a_{n,k}
 &=\left[T_n^g-\sum_{k=0}^{K}g^k b_{n,k}^g\right]\\
 &\quad+\sum_{k=0}^{K}g^k(b_{n,k}^g-a_{n,k})+g^K a_{n,K}.
 \end{aligned}
\]
After pairing with $f$, apply \eqref{eq:parent-remainder} at
truncation index $K+1$ to the first bracket. It is
$O(g^{K+1}(1+t^2)^{r_{n,K+1}})=o(g^K)$.
Every term of the finite coefficient-difference sum is $o(g^K)$
by \eqref{eq:flat-heavy-error}. The last term is $O(g^K)$, since
$a_{n,K}$ is a fixed tempered distribution. This proves
\eqref{eq:matching}. The OS properties and exact normalization
hold by scalar restriction at each $g>0$; the nonzero leading
connected four-point coefficient ensures interaction. The same
argument is uniform on bounded sets of Schwartz tests at each fixed
correlation and truncation order.
\end{proof}

Once the path is selected, $z(g)$ and the heavy masses are prescribed
functions of the sole variable label $g$ of $S_n(g)$, as in
Section~\ref{sec:question}. Scalar-accessible heavy thresholds---the
timelike energies where channels involving heavy states open---can
give continuous spectral weight or branch cuts at positive $g$.
These physical contributions differ from the local heavy term absorbed
in the quartic subtraction in Section~\ref{subsec:parent-checks};
the fixed-test expansion permits them. The exact reference value
$g=\bar g(\mu;g)$ is assumed in the theorem. Point-response existence
and regularity at other scales require the further checks of
Section~\ref{sec:poles}.

\subsection{What must be checked in the parent?}
\label{subsec:parent-checks}

The two estimates address different issues: agreement of perturbative
coefficients, and control of the error after truncating the parent's
perturbative expansion.
Parent existence is assumed in the theorem; both estimates still need
to be established for the proposed parent.

The coefficient estimate \eqref{eq:coefficient-decoupling} compares the
parent scalar coefficients $b_{n,k}$ with the quartic coefficients
$a_{n,k}$ after imposing the \emph{same} physical scalar mass, unit
pole residue and four-point subtraction value. This matching is essential.
Heavy particles can produce large mass corrections, of size $gM_H^2$,
as well as shifts in field normalization and the quartic coupling.
If these conditions can be met within the allowed parent family,
its mass and interaction parameters and the scalar field normalization
must be adjusted. This keeps the physical pole mass, pole residue and
quartic subtraction value at their prescribed values.
Decoupling concerns the remaining difference,
after the light mass and interaction strength have been fixed.

A heavy fermion box illustrates why this is plausible. Its local
four-scalar contribution is absorbed by matching the quartic coupling
at the prescribed subtraction point. With external momenta, the light
mass $m$ and subtraction scale $\mu$ held fixed as $M_H$ grows,
the remaining momentum dependence is expected to be suppressed by
inverse powers of $M_H$.
The $O(N)$ symmetry includes $\Phi_a\mapsto-\Phi_a$, which forbids
odd scalar terms. After matching the mass, kinetic and quartic terms,
the first expected residual scalar operators have dimension six.
In one-component notation, examples are $\phi^6/M_H^2$ and
$\phi^2(\partial\phi)^2/M_H^2$.
This local expansion motivates \eqref{eq:coefficient-decoupling},
in the low-momentum decoupling setting of
\citet{AppelquistCarazzone1975}, where heavy effects remain in the
renormalization of light parameters.
Proving it for every Schwartz probe additionally requires control at
large external momenta, mixed heavy/light graphs, contact extensions,
and conversion to the physical pole mass and residue. That decoupling
motivation also supplies no bound on the exact growing-mass remainder.

Where the joint prescription \eqref{eq:parent-total-degree} applies,
substituting $u_i=gz_i$ makes all additional interaction couplings
vanish with $g$. At fixed consistent ratios and mass data,
$b_{n,k}(z,M_H)$ then collects the terms of total coupling degree $k$.
Otherwise it denotes the coefficient of order $k$ in the fixed-ratio
prescription \eqref{eq:parent-formal-series}. The chosen family
subsequently evaluates these coefficients at $z(g),M_H(g)$.

The remainder estimate \eqref{eq:parent-remainder} instead concerns
the actual scalar distribution $T_n(g,z,M_H)$ of the parent, including
its non-perturbative contributions. For truncation after degree $K-1$,
define the remaining distribution divided by $g^K$:
\begin{equation}
 F_{n,K}(g,z,M_H):=g^{-K}
 \left[T_n(g,z,M_H)-\sum_{k=0}^{K-1}g^k b_{n,k}(z,M_H)\right].
 \label{eq:normalized-parent-tail}
\end{equation}
If a parent family exists at fixed allowed $z,M_H$ and has an all-order
asymptotic expansion, each fixed smearing satisfies
$\langle F_{n,K}(g,z,M_H),f\rangle\to
\langle b_{n,K}(z,M_H),f\rangle$ as $g\downarrow0$.
This does not control how the remainder grows when $M_H$ also tends
to infinity. On the \emph{same} chosen family as the coefficient
estimate, the theorem requires
\[
 |\langle F_{n,K}(g,z(g),M_H(g)),f\rangle|
 \leq C p(f)\ell_H(g)^r,
 \qquad \ell_H(g)=1+\log\frac{M_H(g)}{m_*}.
\]
Here $C$ and the continuous Schwartz seminorm $p$ are independent of
$g$ along the selected family; the finite exponent $r$ may depend on
$n,K$. This is the needed control on the actual error after truncation.

A stronger sufficient condition concerns only the selected family.
For fixed $n,K$, choose a small $g_0>0$, independent of $f$.
Suppose for every fixed Schwartz test $f$, the function
$g\mapsto\langle F_{n,K}(g,z(g),M_H(g)),f\rangle$ extends continuously
to $[0,g_0]$. Each such function is then bounded. Uniform boundedness
for tempered distributions supplies one continuous Schwartz seminorm
$p$ and a constant $C$ independent of $g$, giving
\eqref{eq:parent-remainder} with $r=0$.
The theorem permits finite logarithmic growth along the selected path.
Even if Borel summation recovers the exact parent correlations at each
fixed mass, it does not by itself give the required bound along the
path of increasing masses. The bound could instead be proved directly
for the exact correlations.

\subsection{Why a completely asymptotically free parent?}
\label{subsec:caf-parent}

Complete asymptotic freedom gives a reason to investigate this route:
all dimensionless interaction couplings approach zero as the probing
energy tends to infinity. Weak interactions at short distances offer
a starting point for renormalization-group estimates. If the parent
also has a stable massive quantum phase, decay of correlations at
large distances could support cluster expansions, which control how
distant regions contribute to correlations. Together these features
suggest control of both the ultraviolet and infrared regimes, with
no expected ultraviolet Landau growth on the selected trajectory.
They motivate a construction programme, rather than prove existence.
Ba\l{}aban's gauge-theory work supplies relevant ultraviolet methods
\cite{Balaban1987}; even the full four-dimensional Yang--Mills
existence and mass-gap problem remains open \cite{JaffeWitten2000}.

For this mechanism to help, the additional particles must start
affecting the running while the scalar interaction is still weak.
Their masses must also grow fast enough as $g\downarrow0$ to make
inverse-mass corrections flat after fixed Schwartz smearing.
We compare these requirements using the mass path
\eqref{eq:heavy-path}. Here $Q_0$ and $m_*$ are fixed reference
energies. In each selected theory, $g$ is the fixed input coupling
at $Q_0$, while $Q$ varies to probe its energy dependence. Taking
$g\downarrow0$ instead selects different theories, with different
heavy masses $M_H(g)$.

For this scale comparison, use the ultraviolet form of the scalar
one-loop beta function, omitting the light-mass threshold dependence.
Let $b_\phi>0$ be its coefficient for the scale variable $\log Q$,
and let $\bar g_\phi^{(1)}(Q;g)$ denote the solution with reference
value $g$ at $Q_0$. This auxiliary flow gives
\[
 \frac{d\bar g_\phi^{(1)}}{d\log Q}
       =b_\phi(\bar g_\phi^{(1)})^2,
 \qquad
 \bar g_\phi^{(1)}(Q_0;g)=g,
 \qquad
 \bar g_\phi^{(1)}(Q;g)
       =\frac{g}{1-b_\phi g\log(Q/Q_0)}.
\]
Its formal pole is
$\Lambda_L^{\rm formal}(g)=Q_0\exp[1/(b_\phi g)]$.
This formula extrapolates the scalar one-loop equation without
additional particles; it asserts no exact singularity of the parent.

The chosen heavy masses grow much more slowly than this comparison
scale, although their inverse powers are flat in $g$. To verify the
separation, put $L=\log(1/g)$, so $g=e^{-L}$, and let
$c=\log(m_*/Q_0)$, a fixed constant. Then
\begin{equation}
 \begin{aligned}
 \log\frac{M_H(g)}{m_*}&=L^2=o(1/g),\\
 \log\frac{\Lambda_L^{\rm formal}(g)}{Q_0}
       &=\frac1{b_\phi g},\\
 g\log\frac{M_H(g)}{Q_0}
       &=e^{-L}(L^2+c)\longrightarrow0.
 \end{aligned}
 \label{eq:parent-threshold-scales}
\end{equation}
The first and third statements follow because an exponential beats
any fixed power of $L$. In fact,
$\log(M_H/\Lambda_L^{\rm formal})=c+L^2-e^L/b_\phi\to-\infty$,
so $M_H/\Lambda_L^{\rm formal}\to0$. The additional particles
therefore enter far below the scale where the scalar comparison
would diverge.

The same calculation bounds the scalar one-loop coupling throughout
the interval below the heavy scale. For small enough $g$,
$M_H(g)>Q_0$ and
$\epsilon(g):=b_\phi g\log(M_H(g)/Q_0)\to0$. Once
$\epsilon(g)\leq1/2$,
\[
 g\leq\bar g_\phi^{(1)}(Q;g)
       \leq\frac{g}{1-\epsilon(g)}\leq2g,
 \qquad Q_0\leq Q\leq M_H(g).
\]
At the upper end, $\bar g_\phi^{(1)}(M_H(g);g)/g\to1$.
These are exact inequalities for the displayed one-loop comparison.
An actual parent estimate requires control of matching corrections
and of the exact correlations; the calculation does not establish
\eqref{eq:parent-remainder}.

Above the heavy scale, gauge, Yukawa and additional scalar
interactions change the coupled beta functions and can permit all
interactions to decrease together \cite{GiudiceEtAl2015}.
Early coupled gauge/scalar analyses found restrictions on simultaneous
asymptotic freedom and complete Higgs breaking in the representation
classes examined \cite{ChengEichtenLi1974}.
CAF describes this energy evolution as $Q\to\infty$ in each fixed
theory. The supplied ratios $z_i(g)=u_i(g)/g$ instead select the
family at fixed $Q_0$. They must come from matching an allowed CAF
trajectory to the physical normalization, including the required
corrections; the scalar comparison does not choose them.

To make these questions model-specific, consider a candidate for
$N=1$ with an $SU(2)$ gauge field, the retained real gauge-singlet
scalar $\phi$, a Higgs doublet $H$, ten fundamental Dirac fermions and
one Dirac singlet. Write $e(Q)$ for its running gauge coupling.
Above the heavy thresholds its one-loop gauge equation is
\[
 16\pi^2\frac{d e^2(Q)}{d\log Q}=-e^4(Q).
\]
The coupled one-loop equations admit a CAF ray: an energy trajectory
on which the nonzero Yukawa squares and scalar quartics stay in
fixed positive ratios to $e^2(Q)$. The mixed scalar coupling
$\delta$ for the $\phi^2H^\dagger H$ term is zero on this ray.
Since $e^2(Q)\to0$ as $Q\to\infty$, all interactions tend to zero
together. These are running Lagrangian couplings at $Q$; $z_i(g)$
instead uses the physical input labels at $Q_0$.

Selecting this ray requires tuning: small deviations in the running
ratios grow as $Q$ increases. The supplied couplings must therefore
be chosen to follow a CAF trajectory available in the parent theory,
with corrections included when matching it to the physical
normalization at $Q_0$.
Nearby arbitrary choices need not be CAF. Ultraviolet attraction
would make small ratio deviations shrink as energy increases;
this stronger property is studied in \citet{HoldomRenZhang2015}.

The candidate also allows a massive tree spectrum. At a stable
tree Higgs vacuum, the mass parameters can be chosen so that the
vectors and radial Higgs (the massive fluctuation of the Higgs magnitude)
have mass $M_H$, the additional fermions
have masses between $M_H$ and $\sqrt5M_H$, and the retained scalar
has the fixed mass $m>0$. This checks the classical mass pattern;
it does not prove a stable quantum phase or its physical mass gap.
Appendix~\ref{app:parent} gives the Lagrangian, full one-loop equations,
ratio analysis and tree-mass calculation.

Turning these checks into the parent required by
Theorem~\ref{thm:parent} remains the task. One must find the corrected
CAF trajectory beyond one loop, construct the normalized continuum
family, and prove a stable quantum phase with a physical mass gap.
One must then establish both the coefficient comparison
\eqref{eq:coefficient-decoupling} and the exact remainder estimate
\eqref{eq:parent-remainder} along the same growing-mass path.
Neither follows from the one-loop ray or tree masses. If the theorem's
premises are met, its scalar hierarchy is non-Gaussian and therefore
cannot be a limit in a class whose limits are all Gaussian, as in
Section~\ref{sec:cutoff}. Gauge and fermion field content and a one-loop
ray alone do not establish that continuum conclusion. The candidate
makes these questions concrete for one scalar; an $O(N)$ parent with
$N>1$ requires its own field-content and flow analysis.

\section{The conjecture and the next construction tests}
\label{sec:next}

The usual lattice, Borel and perturbative Landau-pole arguments do not,
by themselves, exclude Conjecture~\ref{conj:existence}. It asks for a
massive, physically normalized OS family with every finite-$N$ quartic
coefficient. The results here identify mechanisms available to a
construction and make their next tests specific; the interacting
family remains to be constructed.

The function example separates vanishing limits along every permitted
bare trajectory from a nonzero completion obtained from renormalized
coefficients. A flat correction also gives it regular exact running.
The collective calculation recovers distributions after ordinary Borel
summation fails and constructs two positive spectral completions with
the same massless leading-$1/N$ collective coefficients, in the stated
contact convention. These are generalized free fields: their connected
correlations above order two vanish. They illustrate completion and
positivity, but do not construct the interacting target.

The parent candidate gives a concrete programme for the interacting
construction. A first test is to calculate its scalar coefficients at
low perturbative order and compare them with the quartic coefficients,
after matching the light pole mass, unit pole residue and exact quartic
subtraction value. This model-specific calculation remains to be done.
The next task is the full coefficient estimate
\eqref{eq:coefficient-decoupling}, covering every Schwartz test,
including large external momenta and the chosen contact extensions.

An actual normalized continuum parent must
also supply a family on a corrected CAF trajectory compatible with
the prescribed physical normalization, a stable massive quantum phase,
and the exact remainder estimate
\eqref{eq:parent-remainder}. Both estimates
must hold along the same selected ratios $z(g)$ and growing masses
$M_H(g)$: control at each fixed mass does not bound the error along
this path. The checked one-loop CAF ray and tree spectrum identify a
candidate for these tasks. Theorem~\ref{thm:parent} states why an actual
parent with the two estimates would make its unchanged scalar
correlations realize the conjecture for $N=1$.

We prefer a non-perturbative completion with a regular physical
response. Establishing point responses at other scales requires the
additional checks of Section~\ref{sec:poles}; this regularity is separate
from Theorem~\ref{thm:parent}'s distributional conclusion.

\section*{Acknowledgement}
The author thanks Arkady Tseytlin, his supervisor for this article,
for guidance, comments and encouragement to prepare it.
This work develops the author's dissertation,
\emph{$\phi^4$ triviality, Landau pole and beyond-Borel summability:
an axiomatic QFT point of view} (Imperial College London, 2026).
AI systems assisted research, calculations, drafting and review under
the author's direction. The Borel--Fourier method for arbitrary
Schwartz tests was found with AI while exploring accelerated summation.
The positive spectral realization of the repaired logarithmic flow
originated in AI-assisted review; AI-assisted calculations supplied
its representation and estimates. The author is responsible for the content.

\appendix

\setcounter{section}{0}
\section{Why does the summation work on every Schwartz test?}
\label{app:distributional}

Can one sum the logarithmic coefficients by one continuous rule for
all Schwartz probes, including those nonzero at $w=0$ or reaching the
pole? The positive-ray Borel integral can fail even though the
geometric kernel is known. Here the two outward imaginary rays give
well-defined distributions, and their mean gives its principal value.
The estimates below prove Theorem~\ref{thm:bf}.

Fix $\mu>0$, $A=8\pi^2$ and $w=P^2\geq0$. A half-line Schwartz
function is smooth up to zero and has finite seminorms
$p_{m,n}(f)=\sup_{w\geq0}(1+w)^m|f^{(n)}(w)|$ for all $m,n$,
in fixed momentum units. The series is
\begin{equation}
 \widehat W_f(g)=-A\sum_{j\geq0}c_j[f]g^{j+1},\qquad
 c_j[f]=\int_0^\infty[\log(w/\mu^2)]^j f(w)\,dw.
 \label{eq:app-a-series}
\end{equation}
The coefficient integrals exist at both endpoints. No assumption
$f(0)=0$ is imposed.

\paragraph{The logarithmic profile and its Borel transform.}
Include the Jacobian $dw/dt$ in the profile and fix the Fourier convention:
\begin{equation}
 t=\log(w/\mu^2),\qquad h_f(t)=\mu^2e^t f(\mu^2e^t),\qquad
 \widehat h_f(s)=\int_{\mathbb R}h_f(t)e^{-ist}\,dt.
 \label{eq:app-a-profile}
\end{equation}
Thus $c_j[f]=\int t^j h_f(t)\,dt$. Every $t$ derivative is a finite
sum of $w^{j+1}f^{(j)}(w)$: these are $O(e^t)$ as $t\to-\infty$
even when $f(0)\ne0$, and Schwartz decay makes them decrease faster
than any fixed exponential as $t\to+\infty$. This also bounds each
profile seminorm by finitely many probe seminorms, so $f\mapsto h_f$
is continuous into $\Sch(\mathbb R)$. The useful low derivatives are
\begin{equation}
 h_f'=wf+w^2f',\qquad h_f''=wf+3w^2f'+w^3f'',\qquad w=\mu^2e^t.
 \label{eq:app-a-profile-derivatives}
\end{equation}
For $0<\delta<1$, the weight in
$M_\delta(f):=\int e^{\delta|t|}|h_f(t)|\,dt$ behaves as
$w^{-\delta}$ near zero and $w^\delta$ at infinity. It is integrable
without endpoint vanishing, and
\begin{equation}
 |c_j[f]|\leq j!\delta^{-j}M_\delta(f),\qquad
 M_\delta(f)\leq C_{\delta,\mu}p_{3,0}(f).
 \label{eq:app-a-moment-bound}
\end{equation}
Dominated convergence identifies the genuine local Borel series with
\begin{equation}
 B_f(\xi):=\mathcal B[\widehat W_f/g](\xi)
 =-A\sum_{j\geq0}\frac{c_j[f]}{j!}\xi^j
 =-A\int_{\mathbb R}h_f(t)e^{\xi t}\,dt,\qquad |\xi|<1.
 \label{eq:app-a-borel-identity}
\end{equation}
The integral continues it holomorphically to $\operatorname{Re}\xi>-1$:
the negative endpoint has decay $e^{(1+\operatorname{Re}\xi)t}$,
and Schwartz decay controls the positive endpoint and all required
$\xi$ derivatives. In particular,
$B_f(-is)=-A\widehat h_f(s)$ and $B_f(is)=-A\widehat h_f(-s)$.

\paragraph{Abel evaluation fixes the two signs.}
Put $r=1/g$ and $a=\mu^2e^r$, with $g>0$. The outward rays
$\xi=-is$ and $\xi=is$ have $d\xi=-i\,ds$ and $i\,ds$, giving
\begin{equation}
 I_\downarrow=iA\int_0^\infty e^{irs}\widehat h_f(s)\,ds,
 \qquad I_\uparrow=-iA\int_0^\infty e^{-irs}\widehat h_f(-s)\,ds.
 \label{eq:app-a-rays}
\end{equation}
These integrals converge absolutely by Fourier decay. Insert
$e^{-\epsilon s}$, $\epsilon>0$, to justify Fubini; then
\begin{equation}
 I_{\downarrow,\epsilon}=A\int_{\mathbb R}
       \frac{h_f(t)}{t-r-i\epsilon}\,dt,
 \qquad I_{\uparrow,\epsilon}=A\int_{\mathbb R}
       \frac{h_f(t)}{t-r+i\epsilon}\,dt.
 \label{eq:app-a-abel}
\end{equation}
For the Cauchy transform $C_f(z)=A\int h_f(t)/(t-z)\,dt$, $z\notin\mathbb R$,
these approach $C_f(r+i0)$ and $C_f(r-i0)$, respectively. The
identity $(x\mp i0)^{-1}=\PV(1/x)\pm i\pi\delta(x)$ follows by
subtracting the probe's value at zero in the real part and using
$\epsilon/(x^2+\epsilon^2)$ as a delta approximation of mass $\pi$.
Since $h_f(r)=af(a)$, it gives
\begin{equation}
 I_\downarrow=\langle W^0,f\rangle+i\pi Aa f(a),\qquad
 I_\uparrow=\langle W^0,f\rangle-i\pi Aa f(a),\qquad
 W^0=A\PV\frac1{\log(w/a)}.
 \label{eq:app-a-boundaries}
\end{equation}
The principal value in $w$ agrees with the one in $t$. Equal exclusions
$\eta$ about $t=r$ give distances $a(1-e^{-\eta})$ and
$a(e^\eta-1)$ from $w=a$, whose ratio is $e^\eta\to1$.
The local singular part is $Aa f(a)/(w-a)$; the possible difference
is its coefficient times the logarithm of that ratio, tending to
zero. There is no additional term supported at the pole.
The arithmetic mean therefore equals $\langle W^0,f\rangle$ and
is real for real probes. The difference is $2\pi iAa f(a)$.
Because $|af(a)|\leq p_{J,0}(f)a(1+a)^{-J}$, choosing $J>L+1$
makes this ambiguity $O(e^{-L/g})$ for every fixed $L>0$,
uniformly on bounded Schwartz sets. The decay-weight order $J$ is
chosen for the requested rate $L$; bounded sets of the full Schwartz
space bound each such seminorm.

\paragraph{Continuity and the finite-order error.}
Two integrations by parts in the Fourier transform, splitting
$|s|\leq1$ from $|s|>1$, give
\begin{equation}
 \|\widehat H\|_1\leq2\|H\|_1+2\|H''\|_1,
 \qquad H\in\Sch(\mathbb R).
 \label{eq:app-a-fourier-lone}
\end{equation}
Set $q_4(f)=\max_{0\leq j\leq2}p_{4,j}(f)$, as in
Theorem~\ref{thm:bf}. With $dt=dw/w$,
\eqref{eq:app-a-profile-derivatives} bounds the right side for $H=h_f$
by $C_\mu q_4(f)$. Thus $|I(f,g)|\leq C_\mu A q_4(f)$ for either
ray and their mean: each is a continuous distribution on the full
half-line space.

For each $K\geq1$, $K+1$ integrations by parts in
\eqref{eq:app-a-rays} use
$\widehat h_f^{(j)}(0)=(-i)^j c_j[f]$; all boundary terms at infinity
vanish. Keeping the first $K$ coefficients and placing the next one
in the error yields, for either ray and their mean,
\begin{equation}
 \left|I(f,g)+A\sum_{j=0}^{K-1}c_j[f]g^{j+1}\right|
 \leq A g^{K+1}\left(|c_K[f]|+\|\widehat h_f^{(K+1)}\|_1\right)
 \leq C_{K,\mu}A g^{K+1}q_4(f).
 \label{eq:app-a-remainder}
\end{equation}
High Fourier derivatives come from multiplication by $t$:
$\widehat h_f^{(j)}=\widehat{(-it)^j h_f}$. Applying
\eqref{eq:app-a-fourier-lone} to $H=(-it)^j h_f$ requires only
$h_f,h_f',h_f''$. In $w$ variables their integrals contain logarithmic
powers multiplying $f,wf'$ or $w^2f''$, all controlled by $q_4(f)$:
for each fixed nonnegative integer $m$,
\[
 \int_0^\infty(1+|\log(w/\mu^2)|)^m
       \frac{1+w+w^2}{(1+w)^4}\,dw<\infty.
\]
Logarithmic powers are integrable at zero, and the tail is bounded by
a constant times $(1+\log w)^m/w^2$. The constant may depend on
$K$, but the weight order remains four and the derivative order two.
This fixed seminorm controls each finite-order error; arbitrarily fast
exponential estimates for the ray difference use the weight chosen above.
This proves the complete expansion uniformly on bounded probe sets;
it gives no uniform bound as packet widths shrink and their derivatives
diverge.

\paragraph{Four-dimensional tests and the scope of the prescription.}
For an arbitrary $\varphi\in\Sch(\mathbb R^4)$, use
\begin{equation}
 (\mathcal R\varphi)(w)=\frac w2\int_{S^3}
       \varphi(\sqrt w\,\omega)\,d\omega,
 \qquad \langle W(P^2),\varphi\rangle=\langle W,\mathcal R\varphi\rangle,
 \label{eq:app-a-radial}
\end{equation}
with unnormalised sphere measure. The factor $w/2$ is the radial
Jacobian. Angular averaging is smooth and even in $\sqrt w$; its
even Taylor expansions give smoothness at $w=0$, and Schwartz
bounds give continuity into the half-line space. Its image vanishes
at zero, while the preceding proof covers the larger half-line class.
The bounds therefore transfer to every momentum test and, by the
continuous inverse Fourier map with convention
$(2\pi)^{-4}\int e^{iP\cdot x}W(P^2)\,d^4P$, to position space.

\paragraph{Relation to other summation prescriptions.}
The two exact imaginary rays use the original Borel transform
without a contour rotation. A Stokes direction is a Borel ray whose
singularities require choosing a continuation above or below it.
The lateral operators $\mathcal S_{0+}$ and $\mathcal S_{0-}$ sum
formal data on those two sides; the Stokes operator $\mathcal G_0$
records the change between them through
$\mathcal S_{0-}=\mathcal S_{0+}\mathcal G_0$. Where these resurgent
operations exist, the general median is
\[
 \mathcal S_{0,\mathrm{med}}=\mathcal S_{0+}\mathcal G_0^{1/2}
                         =\mathcal S_{0-}\mathcal G_0^{-1/2}.
\]
The half-power uses this change between lateral sums; the median
generally differs from an arithmetic mean \cite{Dorigoni2019}.
Acceleration changes the
transform and recovery kernel to control growth that defeats ordinary
Laplace recovery \cite{BellonClavier2019}. Our continuous linear rule
is defined at positive real $g$; we have not proved analytic summability
in an open angular sector of the complex $g$ plane \cite{Sauzin2014Summability}.
Averaging Wick products formed from the two ray values need not give
Wick products of their mean; Wick correlations are formed after
choosing the two-point kernel. Further flat terms can share the
same expansion, so no unique completion is selected.

\setcounter{section}{1}
\section{Positive spectra and preservation of the collective expansion}
\label{app:positive}

This appendix gives the spectral, flatness and exact reference checks
for Proposition~\ref{prop:completions}. Positive spectral covariances
$C_D,C_R$ remove the Euclidean pole; the extensions $W_j=C_j-A$ preserve
the massless leading-$1/N$ collective coefficients after contact
subtraction. We compare the two families at one exact reference value,
which still leaves their spectra different.

Use $w=P^2\geq0$ for Euclidean squared momentum, $s\geq0$ for mass
squared, $A=8\pi^2$, a fixed momentum scale $\mu>0$, and $0<g<1$:
\begin{equation}
 a=\mu^2e^{1/g},\qquad b=a/e-\mu^2>0,\qquad
 W^0(w)=A\PV\frac1{\log(w/a)}.
 \label{eq:b-scales}
\end{equation}
Here $e$ is Euler's number; $a,b,w,s$ have units of mass squared.
The PV excludes equal shrinking distances about $w=a$.
For half-line Schwartz tests, smooth up to zero and rapidly decreasing
with every derivative, the prescribed expansion is
$\langle W^0,f\rangle\sim-A\sum_{k\geq0}c_k[f]g^{k+1}$, where
$c_k[f]=\int_0^\infty[\log(w/\mu^2)]^k f(w)\,dw$.
Its finite-order remainder is bounded by a continuous Schwartz
seminorm, as stated in \eqref{eq:ray-remainder}; $f(0)$ need not vanish.

\subsection{The two positive spectral measures}

Write $z$ for the complex continuation of $w$. In the convergent
spectral representation used here,
\begin{equation}
 C(z)=\int_0^\infty\frac{\rho(ds)}{z+s},\qquad \rho\geq0,
 \qquad \int_0^\infty\frac{\rho(ds)}{\mu^2+s}<\infty.
 \label{eq:b-stieltjes}
\end{equation}
The factors $1/(z+s)$ are the analytic continuations of free Euclidean
propagators of mass $\sqrt s$.
Nonnegative weights give positive state norms. Integrating the
positive-frequency free two-point functions over the same spectral
measure, with $s\geq0$, gives a Minkowski two-point function whose Fourier transform is supported
in the forward energy-momentum cone.
The Euclidean covariance $C$ is reflection positive and Wick pairings give a
generalized free field \cite{Greenberg1961,StreaterWightman2000}.
For the upper cut boundary, the free propagator $1/(z+s)$ has imaginary part
$-\pi\delta$. Hence the continuous density is
$\rho(s)=-\operatorname{Im}C(-s+i0)/\pi$.
A pole at $z=-s_*$ with positive residue $Z_*$ gives the atom
$Z_*\delta(s-s_*)\,ds$. The density is dimensionless here, while
$Z_*$ has units of mass squared.

For the dispersive choice, start with
$F(z)=A/\operatorname{Log}(z/a)$ on the principal branch. Its cut is
$(-\infty,0]$, and
\begin{equation}
 F(-s+i0)=\frac{A}{\log(s/a)+i\pi},\qquad
 \rho_a(s)=\frac{A}{\log^2(s/a)+\pi^2}>0.
 \label{eq:b-d-cut}
\end{equation}
The unwanted pole at $z=a$ has residue $Aa$, since
$\operatorname{Log}(z/a)=(z-a)/a+O((z-a)^2)$. Subtract it:
\begin{equation}
 C_D(z)=F(z)-\frac{Aa}{z-a},\qquad
 C_D(w)=\int_0^\infty\frac{\rho_a(s)}{w+s}\,ds.
 \label{eq:b-d-spectrum}
\end{equation}
The subtraction is real on the cut, preserving its density, and the
pole is removable with $C_D(a)=A/2$. To justify the spectral equality,
apply Cauchy's formula to $C_D(z)/(z-w)$ around the cut: the outer
circle contributes $O(1/\log R)\to0$, and the circle at zero
vanishes because $C_D(z)\to A$ there. No isolated spectral pole
remains. The integral converges even at $w=0$; putting
$t=\log(s/a)$ gives $C_D(0)=A$, so $0<C_D(w)\leq A$.
This is the analytic-perturbation-theory subtraction
\cite{ShirkovSolovtsov1997,MiltonSolovtsov1997}.

For the repaired-flow choice set
\begin{equation}
 C_R(z)=A[1-K(z)^{-1}],\qquad K(z)=1+\operatorname{Log}(1+b/z).
 \label{eq:b-r-analytic}
\end{equation}
Its cut is $[-b,0]$. For $z=-s+i0$, $0<s<b$, the argument
$1+b/z$ reaches the negative axis from \emph{below}. Thus, with
$L_s=1+\log(b/s-1)$,
\begin{equation}
 K(-s+i0)=L_s-i\pi,\qquad
 C_R(-s+i0)=A\left[1-\frac{L_s+i\pi}{L_s^2+\pi^2}\right].
 \label{eq:b-r-cut}
\end{equation}
Its negative imaginary part gives a positive density. The equation
$K(z)=0$ requires $1+b/z=e^{-1}$, so its only zero is the simple
one at $z=-s_0$, with $s_0=eb/(e-1)>b$. Since
$K'(z)=-b/[z(z+b)]<0$ at this zero, the residue is
$Z=-A/K'(-s_0)=Aeb/(e-1)^2>0$. Cauchy's formula therefore gives
\begin{equation}
 \begin{aligned}
 C_R(w)&=\frac{Z}{w+s_0}+\int_0^b\frac{\rho_R(s)}{w+s}\,ds,\\
 \rho_R(s)&=\frac{A}{[1+\log(b/s-1)]^2+\pi^2},\quad 0<s<b.
 \end{aligned}
 \label{eq:b-r-spectrum}
\end{equation}
There is no further term: both endpoint circles vanish because
$C_R\to A$, and $C_R(z)=Ab/z+O(z^{-2})$ removes the outer contour.
The density vanishes at the endpoints and its integral converges at
$w=0$. Also $K(w)>1$ for $w>0$, giving $0\leq C_R(w)\leq A$,
$C_R(0)=A$ and no positive-Euclidean pole. Its atom is an isolated
mass $\sqrt{s_0}$ above the upper end of its continuum.

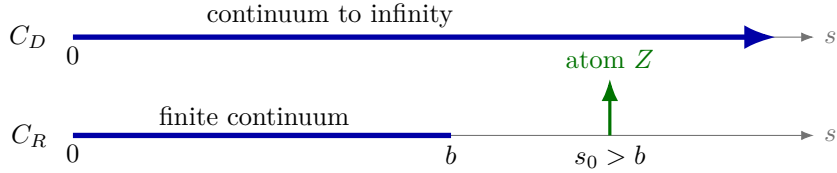
\begin{figure}[htbp]
\centering
\begin{tikzpicture}[>=Latex,x=1cm,y=1cm,font=\small]
 \draw[->,black!55] (1,1.3)--(10.8,1.3) node[right]{$s$};
 \draw[->,black!55] (1,0)--(10.8,0) node[right]{$s$};
 \node[left] at (.8,1.3) {$C_D$}; \node[left] at (.8,0) {$C_R$};
 \draw[blue!65!black,line width=2pt,->] (1,1.3)--(10.3,1.3);
 \draw[blue!65!black,line width=2pt] (1,0)--(6,0);
 \draw[->,green!45!black,line width=1.2pt] (8.1,0)--(8.1,.75);
 \node[above,green!45!black] at (8.1,.75) {atom $Z$};
 \node[below] at (1,1.3) {$0$}; \node[below] at (1,0) {$0$};
 \node[below] at (6,0) {$b$}; \node[below] at (8.1,0) {$s_0>b$};
 \node[above] at (4.4,1.3) {continuum to infinity};
 \node[above] at (3.4,0) {finite continuum};
\end{tikzpicture}
\caption{Mass-squared supports, schematically; the lines show support,
not density height. $C_R$ also has the positive atom at $s_0>b$.}
\label{fig:b-spectral-support}
\end{figure}

\subsection{A common flatness estimate after contact matching}

Positivity belongs to $C_j$; coefficient matching uses $W_j=C_j-A$,
$j=D,R$. This subtraction is the contact
$-A\delta^{(4)}(x-y)$, independent of $g$ and \emph{not flat}.
It changes the coincidence extension. Each OS probe group has support
in $0<t_1<\cdots<t_n$: internal contacts require equal times, while
contacts across reflected groups equate negative and positive times.
Both are excluded, also in mixed
field-order forms. Thus the Wick extensions using $W_j$ reconstruct
the same field as the respective $C_j$, although $W_j$ is not a
positive Gaussian probability covariance. The flat difference is
$W_j-W^0$.

In fixed momentum units use
\begin{equation}
 p_3(f)=\max_{r=0,1}\sup_{w\geq0}(1+w)^3|f^{(r)}(w)|.
 \label{eq:b-test-norm}
\end{equation}
Split the pairing into $[0,a/2]$, $[a/2,3a/2]$ and
$[3a/2,\infty)$, taking $a$ sufficiently large at fixed $\mu$.
Below the pole $|W_D-W^0|=|Aw/(w-a)|\leq2Aw/a$.
For $R$ the needed inequalities are
\begin{equation}
 \begin{aligned}
 W_R&=-A/D_R,\qquad W^0=-A/D_0,\\
 D_R&=1+\log(1+b/w)\geq1,\\
 D_0&=\log(a/w)\geq\log2,\\
 D_R-D_0&=\log[1+e(w-\mu^2)/a],\\
 |D_R-D_0|&\leq C(w+\mu^2)/a\qquad(0<w\leq a/2).
 \end{aligned}
 \label{eq:b-lower-bounds}
\end{equation}
The logarithm's argument stays away from zero, justifying the last
bound. Taking reciprocal differences gives
$|W_R-W^0|\leq CA(w+\mu^2)/a$.
Since both $\int|f|$ and $\int w|f|$ are at most $p_3(f)/2$,
the two lower pieces are $O_\mu(Ap_3(f)/a)$.

Across the pole both cases use
$W^0=Aa\PV(1/(w-a))+Ar(w/a)$, where
$r(x)=1/\log x-1/(x-1)$ is bounded on $[1/2,3/2]$ and
$r(1)=1/2$. Symmetric PV cancels the constant part of $f$, so
\begin{equation}
 \left|Aa\int_{a/2}^{3a/2}\frac{f(w)-f(a)}{w-a}\,dw\right|
 \leq Aa^2\sup_{[a/2,3a/2]}|f'|
 \leq\frac{8A}{a}p_3(f).
 \label{eq:b-middle-bound}
\end{equation}
The bounded remainder and $|W_j|\leq A$ contribute only
$O(Ap_3(f)/a^2)$. Above $3a/2$, $|W^0|\leq A/\log(3/2)$;
Schwartz decay gives $\int_{3a/2}^\infty|f|=O(p_3(f)/a^2)$.
Adding the regions proves
\begin{equation}
 |\langle W_j-W^0,f\rangle|
 \leq\frac{C_\mu A}{a}p_3(f)
 \leq C'_\mu A e^{-1/g}p_3(f),\qquad j=D,R.
 \label{eq:b-common-flatness}
\end{equation}
No zero-endpoint condition was used. The bound is uniform on bounded
sets of tests, giving strong distributional flatness and preserving
every prescribed massless leading-$1/N$ collective two-point coefficient.

For arbitrary four-dimensional momentum tests, use the radial averaging map
\begin{equation}
 (\mathcal R\varphi)(w)=\frac w2\int_{S^3}
          \varphi(\sqrt w\,\omega)\,d\omega,\qquad
 \mathcal R:\Sch(\mathbb R^4)\to\Sch([0,\infty)).
 \label{eq:b-radial-map}
\end{equation}
Angular averaging is even in $\sqrt w$, so the map is smooth at zero;
its seminorms are controlled by finitely many Schwartz seminorms of
$\varphi$. It is therefore continuous. Defining
$\langle W(P^2),\varphi\rangle=\langle W,\mathcal R\varphi\rangle$
transfers the bound to $\Sch'(\mathbb R^4)$, and continuous inverse
Fourier transformation transfers it to position space. Each Wick pairing
is a tensor product over disjoint insertion pairs. At each fixed Wick
order, a finite telescoping sum places one flat two-point difference in
every changed term. The remainder and flatness estimates give uniform
continuity bounds at small $g$ for the other factors.
This gives the joint-test matching in Proposition~\ref{prop:completions};
the reconstruction and joint-test details are in Appendix~\ref{app:gaussian}.

\subsection{A common exact reference value}

At the reference point,
\begin{equation}
 -W_R(\mu^2;g)/A=g,\qquad -W_D(\mu^2;g)/A
 =g_D(g):=g-(e^{1/g}-1)^{-1}.
 \label{eq:b-reference-values}
\end{equation}
The dispersive inherited label differs from its exact two-point value
by a flat term. Since $g_D'(g)\to1$, choose its local inverse
$\gamma(\eta)$ at a desired small exact value $\eta$:
$g_D(\gamma(\eta))=\eta$, with
$\gamma(\eta)-\eta=O(e^{-1/\eta})$.
Comparing $D$ at $\gamma(\eta)$ with $R$ at $\eta$ gives both
$W(\mu^2)=-A\eta$, but their scales are separately
\begin{equation}
 a_D=\mu^2e^{1/\gamma(\eta)},\qquad
 b_R=\mu^2(e^{1/\eta-1}-1).
 \label{eq:b-matched-scales}
\end{equation}
The opening relation $b=a/e-\mu^2$ compares equal inherited labels;
it does not identify these two scales after relabelling. Substituting
a flat label change in each finite truncation preserves all
coefficients; its remainder bound remains valid because
$\gamma(\eta)/\eta\to1$. This selects members of the two families,
with the probing momentum kept fixed.

\begin{table}[htbp]
\centering\small
\caption{The two positive completions. After exact reference matching,
each scale is evaluated at its own label as in
\eqref{eq:b-matched-scales}; matching uses $W_j=C_j-A$.}
\label{tab:b-comparison}
\begin{tabular}{@{}>{\raggedright\arraybackslash}p{.24\linewidth}
 >{\raggedright\arraybackslash}p{.34\linewidth}
 >{\raggedright\arraybackslash}p{.36\linewidth}@{}}
\toprule
Property & Dispersive $D$ & Repaired-flow $R$\\
\midrule
Cut; continuous masses & $(-\infty,0]$; all $s>0$ & $[-b,0]$; $0<s<b$\\[3pt]
Spectral atom & None & $s_0=eb/(e-1)$, $Z=Aeb/(e-1)^2>0$\\[3pt]
Euclidean pole removal & Residue $Aa$ cancelled; $C_D(a)=A/2$ &
 $D_R(w)\geq1$; no pole\\[3pt]
Exact reference & $-W_D(\mu^2)/A=g_D(g)$; flat inverse required &
 $-W_R(\mu^2)/A=g$ already\\[3pt]
Large-$w$ covariance & $A/\log(w/a)$ & $Ab/w$\\
\bottomrule
\end{tabular}
\end{table}

These different spectra and short-distance behaviors distinguish two
generalized free fields even with the same exact reference value and
every massless leading-$1/N$ collective Wick coefficient. Positivity
constrains a completion but does not choose between these measures; further non-perturbative
information must do so.

\setcounter{section}{2}
\section{Wick hierarchies, contact terms and Euclidean regularity}
\label{app:gaussian}

How does a flat two-point repair preserve every Wick coefficient,
and what can changes at coincidence leave undetermined? Matching
tests the full distributions; OS reconstruction uses a domain that
ignores contacts. We give the whole-hierarchy checks for
Proposition~\ref{prop:completions}, a conditional regularity argument
for the physical subtraction response, then exhibit two flat contact
changes preserving the exact reference coupling.

\subsection{Pairings, joint tests and reconstruction}

For a centered scalar with two-point distribution $U$, the first
nontrivial Wick rule is
\begin{equation}
 S_4=U(x_1-x_2)U(x_3-x_4)+U(x_1-x_3)U(x_2-x_4)
                         +U(x_1-x_4)U(x_2-x_3).
 \label{eq:app-c-four-point}
\end{equation}
Each product acts in disjoint pair variables, so it is a tensor
product. Subtracting the disconnected pairings gives zero; all
connected Wick functions above order two similarly vanish.

For $j=D,R$, the completions of Section~\ref{sec:collective}, define
\[
 \begin{aligned}
 T_j&=\mathcal F^{-1}[C_j(P^2)],&
 U_j&=\mathcal F^{-1}[W_j(P^2)],\\
 V&=\mathcal F^{-1}[W^0(P^2)],& W_j&=C_j-A,\qquad A=8\pi^2.
 \end{aligned}
\]
Here $w=P^2$ is Euclidean squared momentum and $\mu$ the reference
momentum. The comparison has scale $a=\mu^2e^{1/g}$ and kernel
$W^0(w)=A\PV[\log(w/a)]^{-1}$.
The positive covariances $C_j$ have spectra supported on mass squared
$s\geq0$. The inverse Fourier convention is
$(2\pi)^{-4}\int e^{iP\cdot x}\widehat U(P)d^4P$.
Let $S^j,S^0$ be the Wick hierarchies of $U_j,V$, with $S_0=1$
and odd functions zero.

For a pairing $\pi=\{(i_\ell,k_\ell)\}_{\ell=1}^r$, use
$z_\ell=x_{i_\ell}-x_{k_\ell}$, $y_\ell=x_{k_\ell}$ and the
inverse linear map $L_\pi(z,y)=x$. On any joint Schwartz test $F$,
\begin{equation}
 (J_\pi F)(z)=\int F(L_\pi(z,y))dy,\qquad
 \langle S^j_{2r},F\rangle
     =\sum_\pi\langle U_j^{\otimes r},J_\pi F\rangle.
 \label{eq:app-c-joint-wick}
\end{equation}
Schwartz decay controls the $y$ integral and its $z$ derivatives,
making $J_\pi$ continuous. Thus the rule includes tests correlating
all insertion coordinates.

The two-point estimates used in this Wick-hierarchy comparison are
\begin{equation}
 |\langle U_j-V,\psi\rangle|\leq K_\mu e^{-1/g}q(\psi),\qquad
 |\langle U_j,\psi\rangle|+|\langle V,\psi\rangle|\leq B_\mu q(\psi),
 \label{eq:app-c-covariance-bounds}
\end{equation}
with one continuous Schwartz seminorm $q$, uniformly at small $g$.
The half-line bounds \eqref{eq:dispersive-flat}, \eqref{eq:rg-flat}
and \eqref{eq:ray-remainder} give these after continuous radial
transfer and Fourier transformation. Telescoping gives
\begin{equation}
 U_j^{\otimes r}-V^{\otimes r}
 =\sum_{\ell=1}^r V^{\otimes(\ell-1)}\otimes(U_j-V)
                          \otimes U_j^{\otimes(r-\ell)}.
 \label{eq:app-c-telescoping}
\end{equation}
Each term contains one flat factor and bounded remaining factors.
Their tensor estimates and $J_\pi$ therefore give
\begin{equation}
 |\langle S^j_{2r}-S^0_{2r},F\rangle|
       \leq C_{r,\mu}e^{-1/g}q_r(F)
 \label{eq:app-c-flatness}
\end{equation}
for a continuous joint seminorm $q_r$. Expanding the factors gives
every massless leading-$1/N$ collective Wick coefficient, uniformly
on bounded joint-test sets at each fixed $r$; the constants and seminorm
orders may depend on $r$.

Comparing the Wick hierarchies of $U_j$ and $T_j$, every changed term
contains a contact because $U_j=T_j-A\delta^{(4)}$. They vanish on
\[
 \Sch_{n,0}=\{F\in\Sch((\mathbb R^4)^n):
       \partial^\alpha F|_{\Delta_n}=0\ \forall\alpha\},\qquad
 \Delta_n=\bigcup_{i<k}\{x_i=x_k\}.
\]
This union includes partial coincidences
\cite[Section~3, Remark~1]{OsterwalderSchrader1973}.
For reflection positivity, write
\begin{equation}
 \sum_{m,n}\langle S_{m+n},f_m^{\theta *}\otimes f_n\rangle,
 \qquad f_m^{\theta *}(x)=
       \overline{f_m(\theta x_m,\ldots,\theta x_1)},
 \label{eq:app-c-mixed-reflection}
\end{equation}
where $\theta(t,\mathbf x)=(-t,\mathbf x)$ and each $f_n$ has
support in $0<t_1<\cdots<t_n$. Ordering excludes internal contacts;
negative reflected times cannot coincide with positive times.
Every $(m,n)$ entry, including $m\ne n$ and vacuum entries, is
unchanged. Thus $U_j$ reconstructs the same positive theory as $T_j$.
The nonflat subtraction $C_j\mapsto C_j-A$ chooses the matching
extension; \eqref{eq:app-c-flatness} concerns its difference from $V$.

Growth needs an additional joint-test argument. The bounds
$0\leq C_j\leq A$ and $|W_j|\leq A$ give
$|S_2(f\otimes h)|\leq A\|f\|_2\|h\|_2$ by Parseval.
With the fixed derivative/weight norm
$p(f)=\max_{|\alpha|\leq3}\sup_x(1+|x|)^3|\partial^\alpha f(x)|$
on $\mathbb R^4$, $\|f\|_2\leq c\,p(f)$, hence
\begin{equation}
 |S_{2r}(f_1\otimes\cdots\otimes f_{2r})|
 \leq(2r-1)!!(c_0A)^r\prod_i p(f_i),\qquad
 (2r-1)!!\leq2^r r!.
 \label{eq:app-c-growth}
\end{equation}
This is the factorized-test growth condition $E0''$ of Osterwalder
and Schrader (1975), with one fixed Schwartz order per insertion.
The theorem in their Appendix converts it to the joint-test condition
$E0'$:
$|\langle S_n,F\rangle|\leq\alpha(n!)^\beta\|F\|_{ns}$ for fixed
$s,\alpha,\beta$, where
$\|F\|_k=\max_{|\gamma|\leq k}\sup_X(1+|X|)^k|\partial^\gamma F(X)|$
\cite[Section~IV.1 and Appendix]{OsterwalderSchrader1975}.
For clustering, restrict the Fourier transform of the joint test to
momenta opposite within each pair. This restriction is Schwartz, so
bounded multipliers make each Wick integral absolutely integrable.
Translate one group of insertion points by a large spatial vector.
Any contraction joining the groups contributes a nonzero linear phase
in the independent spatial pair momenta. Riemann--Lebesgue removes
these terms at large separation, leaving the product of within-group
pairings.

\subsection{Regularity and two reference-preserving contact changes}

A convergent nonzero positive scalar spectrum gives
$D(z)=\int_0^\infty\rho(ds)/(s+z)$ and $D(w)>0$ for $w>0$
\cite{Lehmann1954}. For a tempered measure needing ultraviolet
subtractions, choose $r\geq1$ with
$\int(1+s)^{-r-1}\rho(ds)<\infty$ and $w_0>0$. Then
\begin{equation}
 D(z)=P_{r-1}(z)+(-1)^r(z-w_0)^r
          \int_0^\infty\frac{\rho(ds)}{(s+z)(s+w_0)^r},
 \label{eq:app-c-subtracted-spectrum}
\end{equation}
where $P_{r-1}$ is a subtraction polynomial of degree at most $r-1$.
On compact sets
off $(-\infty,0]$, $|s+z|\geq c(1+s)$ controls the integral and its
complex derivatives. Thus $D$ is analytic at every $w>0$, even if
the spectrum reaches zero. Polynomial contacts can nevertheless
create zeros.

\paragraph{A conditional four-point argument.}
The aim is regularity in all independent Euclidean momentum components,
sufficient for the point coupling of Section~\ref{sec:poles}.
Fix one reconstructed theory $\mathcal T_g$. At each field order
needed below, choose once and for all a Schwartz smoothing sequence
on relative insertion coordinates,
converging to the delta distribution in $\Sch'$. Convolve each
Wightman ordering with this same sequence before restricting the
regularized functions to time-ordering regions. For \emph{every}
cycle in Ruelle's construction, a signed sum of these restricted
orderings with cancelling boundaries, assume convergence in $\Sch'$
and Lorentz invariance of the limit. These are his Assumptions~A;
they concern the cycle sums, not separate limits of every restricted
ordering \cite[Section~4, pp.~362--363]{Ruelle1961}.
The time-ordered cycle then defines a limit, which may depend on the
chosen smoothing sequence. Assume that the chosen full Euclidean
two-point and connected four-point Fourier distributions, including
their contact terms, agree near the probing configurations with the
continuations of the corresponding time-ordered limits.
After removing the overall conservation delta, denote the fixed
component projection of the connected four-point function by $H_g$
and the external scalar propagator by $D_{2,g}$.
Removing disconnected terms first avoids their additional
partial-conservation deltas. All four momenta are incoming and sum
to zero.

For complex energies and real spatial momenta, Ruelle's Theorem~3
allows a singularity only if some proper nonempty subset satisfies
\begin{equation}
 \sum_{i\in X}\operatorname{Im}k_i^0=0,\qquad
 \left(\sum_{i\in X}\operatorname{Re}k_i\right)^2\geq0,
 \qquad \varnothing\ne X\subsetneq\{1,2,3,4\}.
 \label{eq:ruelle-partial-momentum}
\end{equation}
The square uses Minkowski signature $(+,-,-,-)$.
At Euclidean momenta $p_i^E=(\omega_i,\mathbf q_i)$ set
$k_i=(i\omega_i,\mathbf q_i)$. Nonexceptionality means
$\sum_{i\in X}p_i^E\ne0$ for every such $X$.
If $\sum_X\omega_i\ne0$, the first condition fails. Otherwise
$\sum_X\mathbf q_i\ne0$, so the real partial momentum has square
$-|\sum_X\mathbf q_i|^2<0$ and the second fails. No subset satisfies
both conditions.

The proof constructs functions on domains in which all independent
four-momentum components are complex (momentum tubes), and continues
them through nonsingular boundary portions (regular faces). It therefore supplies a joint
holomorphic neighbourhood in the independent momenta on
$\sum_i k_i=0$ \cite[Section~5, pp.~364--366]{Ruelle1961}.
This joint regularity controls full momentum packets and the scale
curve, whose spatial momenta also vary with $Q$.
The packets also sample momentum directions away from that curve.
No mass gap is needed for this nonexceptional argument.
If each external $D_{2,g}$ is also analytic and nonzero near its
probing momentum, division by their product preserves the numerator's
joint analyticity. Scale fixed nonexceptional reference
momenta $p_i^*$ by $p_i(Q)=(Q/\mu)p_i^*$. This makes the normalized quotient and its scale
derivative regular at every finite $Q>0$. Normalized shrinking
packets then converge to these values, independently of their shapes
and widths. Continuity alone would suffice for boundedness on compact
$Q$ intervals with nonzero external values.

OS reconstruction and fixed-test coefficient matching do not
establish these assumed cycle limits or the identification of the
chosen extension. The gap in a general time-ordering construction is
noted in \cite[Section~10]{KravchukQiaoRychkov2021}, which identifies
Bogolubov--Logunov--Oksak--Todorov, \emph{General Principles of Quantum
Field Theory}, p.~505, as a reference for this point. Stronger positivity
including coincident insertions and factorial Schwartz growth give
another time-ordering construction
\cite[Theorem~2 and Corollary~3]{EckmannEpstein1979}, without by
themselves verifying Ruelle's particular common regularization.

Both examples below change coincidence extensions of an already
supplied scalar hierarchy. Recomputing the quotient with altered
extensions can change the Euclidean coupling obtained from the packet
quotient. The reconstructed fields and measurements of fixed observables
in the reconstructed theory remain unchanged. Set
$\varepsilon=e^{-1/g}$ and use symmetric momenta with
$p_i^2=3Q^2/4$, $(p_i+p_k)^2=Q^2$, and reference $Q=\mu$.

\paragraph{A contact term creates a propagator zero.}
Let $D_0$ be its propagator with a convergent nonzero positive spectrum.
Put
$w_*=3\mu^2/4$ and
\begin{equation}
 D_\varepsilon(w)=D_0(w)-\varepsilon\frac{w-w_*}{\mu^4}.
 \label{eq:app-c-flat-amputation-zero}
\end{equation}
This is the flat tempered contact
$\varepsilon\mu^{-4}(\Delta+w_*)\delta^{(4)}(x-y)$, with $\Delta$
the Euclidean Laplacian. It preserves every particle pole's mass and
residue and the reference external value. The spectrum gives
$D_0'(w)<0$ and $D_0(w)\to0$ by domination with its integrable
kernel at $w_*$. Consequently $D_\varepsilon$ decreases from a
positive value at $w_*$ to $-\infty$, crossing one simple zero.
If the connected four-point numerator is regular and does not cancel
the external propagator zeros, the quotient in this altered convention
has a pole. A nonzero numerator at the zero configuration suffices;
a denominator zero alone does not establish a pole.
Keep the full
$n$-point extensions for $n\geq3$ fixed. Recomputing the connected
four-point function changes only its disconnected products, whose
Fourier transforms contain $\delta^{(4)}(p_i+p_k)$. These have no
support near nonexceptional configurations. Hence the exact reference
coupling and all OS restrictions remain unchanged.

\paragraph{Partial four-point channel pole.}
Choose fixed $\Lambda_c^2>0$, $\Lambda_c^2\ne\mu^2$, and
\begin{equation}
 \begin{aligned}
 \widehat K(P)&=\mu^{-8}(P^2-\mu^2)\PV\frac1{P^2-\Lambda_c^2},\\
 T_{12|34}(x)&=\delta^{(4)}(x_1-x_2)\delta^{(4)}(x_3-x_4)
                                                    K(x_1-x_3).
 \end{aligned}
 \label{eq:app-c-partial-contact}
\end{equation}
Symmetric exclusion in $P^2$ defines a tempered shell principal value
with nonzero residue. Pairing $T$ means restricting a test to
$(X,X,Y,Y)$ and retaining the free separation $X-Y$.
Add $\varepsilon(T_{12|34}+T_{13|24}+T_{14|23})$ to the full
four-point extension. For component fields, attach the tensor
$\delta_{a_i a_k}\delta_{a_l a_m}$ to each pairing $(ik|lm)$. Real even rotational
$K$ and the three pairings ensure permutation and $O(N)$ symmetry.
Partial-diagonal support makes the addition zero on $\Sch_{4,0}$
and on the stated strictly ordered reflection forms;
its fixed tempered kernel and flat prefactor preserve every coefficient.

Its Fourier transform is
\begin{equation}
 \widehat T_{12|34}(p)=(2\pi)^4\delta^{(4)}\!\left(\sum_i p_i\right)
                                              \widehat K(p_1+p_2).
 \label{eq:app-c-partial-fourier}
\end{equation}
There is no pair-conservation delta: $K$ connects the pair centers.
Every reference channel has $P^2=\mu^2$, where the numerator vanishes,
so the exact subtraction coupling is unchanged. At $Q^2=\Lambda_c^2$, the added reduced Fourier term has a pole.
Dividing by analytic, nonzero external propagators preserves it, provided
the chosen component projection does not annihilate the added term;
the singlet tensor weight is $N^2+2N>0$. A complete-diagonal contact instead has a polynomial
reduced Fourier transform and cannot itself create such a pole.
The pole-bearing partial extension falls outside Ruelle's compatible
time ordering. Both operations leave separated quartic correlations
and fixed physical responses intact. Their convention-dependent poles
are not counterexamples with a physical Euclidean Landau pole in the
same reconstructed QFT.

\setcounter{section}{3}
\section{The massive saddle: channel data, positive spectrum and flat matching}
\label{app:massive}

Does a massive gap solution turn the bubble channel data of
Section~\ref{sec:collective} into a positive collective covariance
without a Euclidean pole?
Romatschke's massive saddle, the self-consistent leading large-$N$
mass solution, does. We derive its positive spectrum and
smeared weak-coupling expansion, distinguishing the calculated
response from a Lagrangian coupling coordinate that can diverge.
The coefficient comparison approaches $g=0$ from negative values
while the generated constituent mass tends to zero, testing collective
completion separately from the positive-$g$, fixed-mass conjecture.

\subsection{Gap cancellation and the collective channel}

Table~\ref{tab:massive-notation} uses the 2023 cutoff convention
\cite{Romatschke2023}. Energies $\nu,\mu,Q,M$ have units of mass;
$w$ and the spectral variable $s$ are squared masses.
\begin{table}[htbp]
\centering
\small
\caption{Massive-branch conventions at the retained order in $1/N$.}
\label{tab:massive-notation}
\begin{tabular}{@{}>{\raggedright\arraybackslash}p{.25\linewidth}>{\raggedright\arraybackslash}p{.69\linewidth}@{}}
\toprule
Symbol & Meaning\\
\midrule
$\lambda_0$, $\lambda_R(\nu)$ & Bare and renormalized Lagrangian
quartic parameters; $h_B=8\lambda_0$ in our convention.\\[3pt]
$M$, $F_M$, $W_M$ & Constituent pole mass, elementary channel
coefficient and collective kernel $W_M=-F_M/4$.\\[3pt]
$g=-u$, $0<u<1$ & Supplied reference value
$g=F_M(\mu^2)/(32\pi^2)$ at fixed energy $\mu$.\\[3pt]
$\nu$, $Q$, $w=Q^2$; $e$ & Renormalization energy, probing energy,
squared channel momentum; $e=\exp(1)$ is Euler's number.\\
\bottomrule
\end{tabular}
\end{table}

At zero renormalized mass parameter, the gap equation has two
self-consistent mass solutions: $M=0$ retains massless constituent
propagation, while $M>0$ generates a constituent scale. The tachyon-free
massive branch and its singlet bound-state channel were analysed earlier
by Abbott, Kang and Schnitzer \cite[Sections~IV--V]{AbbottKangSchnitzer1976}.\footnote{For the conventional real-field continuation with negative bare coupling,
Bardeen and Moshe found vacuum-instability problems, including metastable
decay exponentially suppressed in $N$ \cite{BardeenMoshe1983}.
The positive collective Wick field below satisfies the OS axioms
independently of that continuation. Romatschke's complex field contours
specify a different continuation \cite{Romatschke2023}.}
We use the massive solution, which obeys
\begin{equation}
 \log\frac{M^2}{\nu^2}=\frac{4\pi^2}{\lambda_R(\nu)}.
 \label{eq:massive-gap}
\end{equation}
Thus $M$ is generated rather than an independently fixed light mass.
At fixed $M$, $\lambda_R(\nu)=4\pi^2/\log(M^2/\nu^2)$ diverges
at $\nu=M$; the channel response instead depends on $M$ and $w$.
Define
\begin{equation}
 I_M(w)=\int_0^1\log[1+w t(1-t)/M^2]dt,\qquad
 H_M(w)=1+I_M(w),\qquad A=8\pi^2.
 \label{eq:massive-bubble-function}
\end{equation}
For $w>0$, the integral is
$2\sqrt{(w+4M^2)/w}\operatorname{artanh}\sqrt{w/(w+4M^2)}-2$.
The denominator of the published four-point channel is therefore
\begin{equation}
 D_{\rm ch}(w)=\frac{4\pi^2}{\lambda_R(\nu)}
       -\log\frac{M^2}{\nu^2}-1-I_M(w)=-H_M(w).
 \label{eq:massive-gap-denominator}
\end{equation}
The gap relation cancels the first two terms. The surviving constant
$-1$ fixes the coupling sign and bound-state location.

Romatschke defines $\Gamma_4$ as minus the connected amputated
\emph{elementary} four-point response
\cite[eqs.~(16)--(20)]{Romatschke2023}. Its leading channel gives
\begin{equation}
 \Gamma_4^{(s)}(w)=\frac{F_M(w)}N+O(N^{-2}),\qquad
 F_M(w)=-\frac{32\pi^2}{H_M(w)},\qquad
 W_M(w)=-\frac{F_M(w)}4=\frac{A}{H_M(w)}.
 \label{eq:massive-channel-identity}
\end{equation}
The normalized collective field is the regulated
$\chi_N=h_B[\Phi_{B,a}\Phi_{B,a}-\langle\Phi_{B,a}\Phi_{B,a}\rangle]/(4\sqrt N)$
of \eqref{eq:chi}. The algebra in \eqref{eq:collective-algebra}
gives its leading covariance $C_{B,2}=h_B/4-F_B/4$.
Subtracting the momentum constant $h_B/4$, a local
$\delta^{(4)}(x-y)$ term, leaves $-F_B/4$ and hence $W_M$ after
renormalization. These are two uses of the same leading data.
Here $W_M$ itself will be the positive covariance used for Wick pairings.
The published elementary four-point calculation includes the needed
quadratic auxiliary-field fluctuations. For the normalized collective
field, connected $r$-point fluctuations scale as $N^{(2-r)/2}$ for
$r\geq2$; retaining order $N^0$ therefore leaves a Gaussian collective
hierarchy, with only its two-point connected function.
The bubble chain retains every coupling order within this single
order of $1/N$.

\subsection{Positive spectrum and the bound state}

For $w\geq0$, $H_M(w)\geq1$, so $0<W_M(w)\leq A$ with no
Euclidean pole. Spectral positivity needs a further check. Write $z$
for the complex continuation of $w$ and use the principal logarithm,
whose cut starts at $z=-4M^2$.
On the interval before the cut,
\begin{equation}
 \begin{aligned}
 H'_M(z)&=\int_0^1\frac{t(1-t)}{M^2+z t(1-t)}dt>0,
                       \qquad -4M^2<z<0,\\
 H_M(-4M^2)&=-1,\qquad H_M(0)=1.
 \end{aligned}
 \label{eq:massive-positive-derivative}
\end{equation}
The endpoint value uses
$\int_0^1\log[(2t-1)^2]dt=-2$. Monotonicity gives exactly one
simple zero $-s_B$ with $0<s_B<4M^2$. Off the real axis the
imaginary part of $H_M(z)$ has the strict sign of
$\operatorname{Im}z$, excluding further nonreal zeros.

For $s>4M^2$, the logarithm's argument in $H_M(-s+i0)$ is negative
on an interval of length $v(s)=\sqrt{1-4M^2/s}$. The upper edge thus
has $\operatorname{Im}H_M=\pi v(s)$, and taking its reciprocal gives
$\rho_M(s)=-\pi^{-1}\operatorname{Im}W_M(-s+i0)\geq0$.
The lower boundary values are the complex conjugates of the upper ones.
Cauchy's contour around the cut and pole yields
\begin{equation}
 \begin{aligned}
 W_M(w)&=\frac{Z_B}{w+s_B}
              +\int_{4M^2}^{\infty}\frac{\rho_M(s)}{w+s}ds,\\
 Z_B&=\frac{A}{H'_M(-s_B)}>0,\qquad
 \rho_M(s)=\frac{A\sqrt{1-4M^2/s}}{|H_M(-s+i0)|^2}.
 \end{aligned}
 \label{eq:massive-spectrum}
\end{equation}
There is no outer contour term because $W_M=O(1/\log|z|)$.
The spectral integral converges: its density vanishes at threshold
and decreases as $1/\log^2(s/M^2)$ at infinity. The atom is below
the two-constituent threshold, with mass approximately $1.84M$
\cite{Romatschke2023}.

Positive superposition of free positive-frequency two-point functions
with these weights, followed by Wick pairings, constructs a
\emph{generalized free collective field}
\cite{Greenberg1961,StreaterWightman2000}. The bounded Euclidean
multiplier supplies Gaussian growth bounds with pairing factor
$(2r-1)!!$ at order $2r$; the positive lower spectral mass gives
clustering. Its connected functions above order two vanish.

\subsection{Exact negative-label normalization and flat matching}

The physically normalized leading channel value is
\begin{equation}
 g_{\rm phys}^{\rm lead}(Q;M)=\frac{F_M(Q^2)}{32\pi^2}
             =-\frac{W_M(Q^2)}A=-\frac1{H_M(Q^2)}.
 \label{eq:massive-physical-value}
\end{equation}
It lies between $-1$ and $0$ for $Q>0$; its sign is unrelated to
negative spectral weight. Fix $g=g_{\rm phys}^{\rm lead}(\mu;M)=-u$.
The condition $H_{M(u)}(\mu^2)=u^{-1}$ selects one mass, since
$H_M(\mu^2)$ decreases from infinity to one as $M$ increases.
Separate the endpoint correction by writing
\begin{equation}
 \begin{aligned}
 E_M(w)&=\int_0^1\log[1+M^2/(w t(1-t))]dt,\\
 H_M(w)&=\log(w/M^2)-1+E_M(w),\\
 M(u)^2&=\frac{\mu^2}{e}\exp[-1/u+E_{M(u)}(\mu^2)]
                         \sim\frac{\mu^2}{e}e^{-1/u}.
 \end{aligned}
 \label{eq:massive-normalization}
\end{equation}
Splitting the parameter integral near its endpoints and using
$\log(1+x)\leq x$ away from them gives
$0<E_M(w)\leq C(M^2/w)[1+\log(w/M^2)]$ for $w\geq M^2$.
This proves the mass asymptotic and the exactly normalized identity
\begin{equation}
 H_{M(u)}(w)=u^{-1}+\log(w/\mu^2)
                       +E_{M(u)}(w)-E_{M(u)}(\mu^2).
 \label{eq:massive-reference-denominator}
\end{equation}
Compare with
\begin{equation}
 K_u(w)=A\PV\frac1{u^{-1}+\log(w/\mu^2)},\qquad
 a_u=\mu^2e^{-1/u}.
 \label{eq:massive-comparison}
\end{equation}
Its formal coefficients are the collective coefficients of
\eqref{eq:log-series} with $g=-u$. Its PV pole at $a_u$ shrinks
towards zero, whereas $W_M$ has no such pole.

The comparison must work for every half-line Schwartz test $f$,
smooth up to zero and rapidly decreasing, including $f(0)\ne0$. With
$q_\mu(f)=\|f\|_\infty+\mu^2\|f'\|_\infty+\mu^{-2}\|f\|_1$,
a continuous Schwartz bound, one has
\begin{equation}
 |\langle W_{M(u)}-K_u,f\rangle|
       \leq CA\mu^2e^{-1/(2u)}q_\mu(f).
 \label{eq:massive-flat}
\end{equation}
Here $\mu$ is fixed and $C$ is independent of sufficiently small
$u$ and of $f$. The exponential is smaller than every power of the
positive parameter $u$, which is flatness on this branch.
Choose $d_u=\mu^2e^{-1/(2u)}$, so
$a_u\asymp M(u)^2\ll d_u\ll\mu^2$ as $u\downarrow0$.
The lower region contains the shrinking pole and mass correction;
the upper region has large denominators. Below $d_u$, use $W_M\leq A$ and
subtract $f(a_u)$ in the PV integral on $[0,2a_u]$.
Rescaling $w=a_u x$ bounds that integral by
$C[a_u\|f\|_\infty+a_u^2\|f'\|_\infty]$; the rest has
$\log(w/a_u)\geq\log2$. Above $d_u$, both denominators are at
least $1/(2u)$. The endpoint bound gives their difference as
$O(u^{-1}e^{-1/(2u)})$, so the reciprocal difference is
$O(Au e^{-1/(2u)})$ and its pairing is controlled by $\|f\|_1$.
This proves flatness even though $W_{M(u)}(0)=A$: the region with
an order-one correction shrinks exponentially.

To verify the smeared coefficients themselves, on $w\geq d_u$
use the finite geometric identity with $\ell=\log(w/\mu^2)$:
\[
 \frac1{u^{-1}+\ell}
 =\sum_{j=0}^{K-1}(-\ell)^j u^{j+1}
       +\frac{u^{K+1}(-\ell)^K}{1+u\ell},\qquad
 1+u\ell\geq\frac12.
\]
The remainder is bounded by the integrable logarithmic moment of
$|f|$ times $2u^{K+1}$. Below $d_u$ the PV contribution and the
omitted moments, of size $O(d_u(1+u^{-j}))$, are flat. Hence
\begin{equation}
 \langle W_{M(u)},f\rangle
 \sim A\sum_{j\geq0}u^{j+1}
               \int_0^\infty[-\log(w/\mu^2)]^j f(w)dw.
 \label{eq:massive-jet}
\end{equation}
The expansion is uniform on bounded sets of tests. The continuous
radial map $f(w)=(w/2)\int_{S^3}\varphi(\sqrt w\,\omega)d\omega$
transfers it to arbitrary $\varphi\in\Sch(\mathbb R^4)$, and
Fourier transformation transfers it to position space.

\subsection{Running, family selection and finite-$N$ scope}

Varying $u$ at fixed $\mu$ selects different masses $M(u)$ and
therefore different theories. Energy running instead holds $M$ fixed
while $Q$ varies. For this leading kernel,
\begin{equation}
 \frac{d g_{\rm phys}^{\rm lead}}{d\log Q^2}
       =(g_{\rm phys}^{\rm lead})^2B(Q^2/M^2),\qquad
 B(r)=\int_0^1\frac{r t(1-t)}{1+r t(1-t)}dt.
 \label{eq:massive-beta}
\end{equation}
Here $B$ increases from zero to one. The predicted positive value
$\bar u=-g_{\rm phys}^{\rm lead}$ obeys
$d\bar u/d\log Q^2=-\bar u^2B$; its reference label $u$ stays fixed.
As $Q$ grows, $g_{\rm phys}^{\rm lead}$ increases from its infrared
limit $-1$ towards zero, while $\bar u$ decreases from one towards zero.
Romatschke's refined effective convention in \emph{Life at the Landau pole}
is $g_{R,\rm eff}(Q)=4\pi^2/H_M(e^2Q^2)
=-4\pi^2g_{\rm phys}^{\rm lead}(eQ;M)$
\cite[Section~II.3, eqs.~(47)--(50)]{Romatschke2022}.
Both its sign normalization and probing scale differ. That paper's
renormalization scheme also gives a different named Landau scale;
our comparison uses the common constituent mass $M$.

\begin{table}[htbp]
\centering
\small
\caption{Collective examples versus the interacting target.}
\label{tab:massive-branches}
\begin{tabular}{@{}>{\raggedright\arraybackslash}p{.18\linewidth}>{\raggedright\arraybackslash}p{.17\linewidth}>{\raggedright\arraybackslash}p{.23\linewidth}>{\raggedright\arraybackslash}p{.32\linewidth}@{}}
\toprule
Case & Label limit & Mass condition & Euclidean pole and status\\
\midrule
Massless collective & $g\downarrow0^+$ & $M=0$ &
UV pole removed by positive completion.\\[3pt]
Massive saddle & $g=-u\uparrow0$ &
$M(u)^2\sim{}$\newline$(\mu^2/e)e^{-1/u}$ &
IR comparison pole removed by mass; positive spectrum.\\[3pt]
Conjecture~\ref{conj:existence} & $g\downarrow0^+$ &
Fixed light mass $m>0$, finite $N$ &
Exact interacting elementary hierarchy remains open.\\
\bottomrule
\end{tabular}
\end{table}

Romatschke uses the nonzero elementary four-point coefficient to
motivate interacting finite-$N$ physics, including an $N=4$ Higgs-sector
application \cite[Section~IV]{Romatschke2023}. At fixed component indices
the amplitude is $O(1/N)$; for normalized singlet two-particle states,
the coherent sum over component pairs can retain an order-one response. The bound state records binding in this
pair channel. In the collective Wick field constructed here it is a
spectral particle, while higher connected correlations vanish.
Thus the same leading data support an explicit generalized free
collective realization and a motivation for finite-$N$ interaction.

The pressure and scattering calculations, and the reported small
finite-cutoff $N=1$ complex-contour lattices, do not establish the complete
continuum hierarchy with the conjecture's positive-$g$, fixed-light-mass
normalization. Regularity of an exact interacting point coupling also
requires the shape-independent packet limits and nonsingular amputation
of Section~\ref{sec:poles}: continuous projected four-point values and
continuous nonzero external two-point values near the subtraction
configurations give one sufficient condition.

\setcounter{section}{4}
\section{An interacting pair response with a timelike pole}
\label{app:timelike}

Can a small response select an interacting theory whose prediction at
another timelike momentum has a pole? In this rigorously constructed
example, a stable bound state makes the response's point value diverge
while every fixed smooth packet of compact support crossing it gives a
finite answer. The deliberately unusual coupling convention follows a
composite two-point response to timelike momenta, illustrating the
coexistence discussed in Section~\ref{sec:poles}.
This is a physical bound-state pole in a composite response, not a
Euclidean elementary-four-point Landau pole.

\subsection{A fixed composite operator and an invertible reference response}

In $1+1$ spacetime dimensions, use a fixed free reference mass $m_0=1$
as the mass unit and take
\begin{equation}
 \mathcal L_E=\frac12(\partial\phi)^2+\frac12\phi^2
       +\lambda\bigl({:\!\phi^6\!:}_1-{:\!\phi^4\!:}_1\bigr),
 \qquad 0\leq\lambda\ll1.
 \label{eq:timelike-model}
\end{equation}
The colons denote Wick ordering with respect to the mass-one free
covariance. This subtraction stays fixed. The positive sextic term
stabilizes the attractive quartic interaction. For sufficiently small
positive $\lambda$, the constructed model satisfies the OS axioms
and has nontrivial scattering
\cite[Theorem~1 and p.~11]{EckmannEpsteinFrohlich1976}.
Its physical elementary-particle mass $m_\lambda$ is an output,
with $m_{\lambda=0}=1$; it need not remain equal to the reference mass.

Keep the composite operator $O={:\!\phi^2\!:}_1$ and its normalization fixed in
every theory. This Wick-square operator probes the two-particle channel.
Let $G_\lambda(x)$ be the Fourier transform of its connected time-ordered
two-point function, where $x=E^2-k^2$ is invariant momentum squared, timelike for
$x>0$. Connectedness subtracts the product of the two $O$ expectation
values. Use the
construction's fixed time-ordered extension and normalize its
continuation from Euclidean momentum by
\begin{equation}
 G_0(x)=2B_0(x),\qquad
 B_0(x)=\frac1{4\pi}\int_0^1\frac{dt}{1-xt(1-t)},\quad x<4,
 \qquad U_\lambda(x):=G_\lambda(x)-G_0(x).
 \label{eq:timelike-response}
\end{equation}
The factor two counts the two free cross-pairings between the $O$ insertions.
At and above the free threshold $x=4$, use the boundary value $x+i0$.
The subtracted baseline always has mass one. The interacting
two-particle threshold is instead $4m_\lambda^2$.

The first correction can be checked directly. Write
$C_1(y)=\int d^2p\,e^{ip\cdot y}/[(2\pi)^2(p^2+1)]$ for the free
Euclidean covariance. Expanding the interaction once gives
\begin{equation}
 \langle O(y)O(0)\rangle_{\lambda,c}
 =2C_1(y)^2
   +24\lambda\int d^2z\,C_1(y-z)^2C_1(z)^2+O(\lambda^2).
 \label{eq:timelike-wick}
\end{equation}
The quartic insertion has positive sign because its coefficient in
the Euclidean action is $-\lambda$. Its four legs must attach to the
four elementary fields in the two $O$ insertions: choosing two for each insertion and pairing them gives
$\binom42(2!)^2=24$ contractions. Wick ordering forbids contractions
within the sextic insertion, and the four fields in the two $O$ insertions cannot absorb its
six legs, so it contributes nothing at this order. Fourier
transformation of the convolution in \eqref{eq:timelike-wick} gives
$U_\lambda(x)=24\lambda B_0(x)^2+O(\lambda^2)$ at each fixed $x<4$
for sufficiently small $\lambda$.

At the below-threshold reference point $x=0$, supply the measurement
\begin{equation}
 s:=U_\lambda(0)
   =\frac{3}{2\pi^2}\lambda+O(\lambda^2),\qquad
 \lambda(s)=\frac{2\pi^2}{3}s+O(s^2).
 \label{eq:timelike-label}
\end{equation}
Constructive smoothness applies to Wick-monomial insertions as well
as to elementary fields
\cite[Lemma~6a and Theorem~7]{EckmannEpsteinFrohlich1976}.
Thus $dU_\lambda(0)/d\lambda$ is continuous from the right at zero
and remains positive on a sufficiently small interval. On that
interval the actual reference response has a smooth inverse:
zero selects the free theory and a small positive $s$ selects one
interacting theory. Its prediction is
\begin{equation}
 \bar s(x;s):=U_{\lambda(s)}(x),\qquad \bar s(0;s)=s.
 \label{eq:timelike-prediction}
\end{equation}
Varying $x$ with the supplied label $s$ fixed probes that same theory.
At singular momenta the prediction is a boundary-value distribution.

\subsection{The atom survives free subtraction and finite smearing}

Fix one such positive $s$, hence $\lambda=\lambda(s)>0$. The model
has a stable two-particle bound state below the interacting threshold
\cite{DimockEckmann1976}. The fixed Wick-square operator has a nonzero
vacuum-to-bound-state matrix element
\cite[Theorem~33, pp.~224--226]{Koch1979}. In its positive
mass spectral measure, this bound state is a discrete delta-function
contribution (an atom). This nonzero matrix element gives a strictly
positive weight $R_\lambda$, yielding locally
\begin{equation}
 G_\lambda(x)=H_\lambda(x)
       +\frac{R_\lambda}{P_\lambda-x-i0},\qquad
 R_\lambda>0,\quad 0<P_\lambda<4m_\lambda^2.
 \label{eq:timelike-bound-state}
\end{equation}
Here $P_\lambda$ is the squared bound-state mass. The remaining
response $H_\lambda$ is analytic near this isolated atom: the
two-particle continuum is separated from it by a gap. The atom is a
stable particle contribution, rather than a resonance in the continuum.

The free baseline has only continuous spectral weight,
\begin{equation}
 G_0(x)=\int_4^\infty\frac{\rho_0(\sigma)}{\sigma-x-i0}\,d\sigma,
 \qquad \rho_0(\sigma)=\frac1{\pi\sqrt{\sigma(\sigma-4)}}
                  \quad(\sigma>4).
 \label{eq:timelike-free-spectrum}
\end{equation}
This follows by taking the imaginary part of the parameter integral
in \eqref{eq:timelike-response}. Its threshold density is locally
integrable and has no delta function. Hence subtracting $G_0$ cannot
cancel $R_\lambda\delta(\sigma-P_\lambda)$ in the interacting
spectrum. This remains true if $P_\lambda=4$: the free boundary
value grows only as $|4-x|^{-1/2}$ there, whereas the atom gives a
simple pole. No ordering between the fixed free threshold $4$ and
the interacting threshold $4m_\lambda^2$ is assumed. The continuum
part of $U_\lambda$ need not be nonnegative; its atom retains the
positive weight $R_\lambda$.

The time ordering fixes the pole prescription,
\begin{equation}
 \frac{R_\lambda}{P_\lambda-x-i0}
 =R_\lambda\PV\frac1{P_\lambda-x}
      +i\pi R_\lambda\delta(x-P_\lambda).
 \label{eq:timelike-boundary}
\end{equation}
Choose a fixed positive invariant-mass-squared window $(a,b)$ containing
$P_\lambda$, below $4m_\lambda^2$ and excluding other interacting
singularities.
For $f\in C_c^\infty(a,b)$ its measured response is
\begin{equation}
 \langle U_\lambda,f\rangle
 =\int_a^b[H_\lambda(x)-G_0(x)]f(x)\,dx
   +R_\lambda\PV\int_a^b\frac{f(x)}{P_\lambda-x}\,dx
   +i\pi R_\lambda f(P_\lambda).
 \label{eq:timelike-packet}
\end{equation}
The first integral is finite even if the free threshold lies in the
window. For the second, the identity
\begin{equation}
 \PV\int_a^b\frac{f(x)}{P_\lambda-x}\,dx
 =\int_a^b\frac{f(x)-f(P_\lambda)}{P_\lambda-x}\,dx
   +f(P_\lambda)\log\frac{P_\lambda-a}{b-P_\lambda}
 \label{eq:timelike-pv-subtraction}
\end{equation}
removes the singularity: the integrand is bounded by
$\|f'\|_\infty$. The delta term simply samples the packet at the
bound-state mass squared. Thus each fixed smooth packet, including one
crossing the pole, gives a finite answer.

These invariant packets have a direct spacetime interpretation. Choose
$q\in C_c^\infty(\mathbb R)$ with $\int q(k)\,dk=1$ and set
$\Phi(E,k)=2E f(E^2-k^2)q(k)$ on $E>0$, extending it by zero to
$E\leq0$. Since $a>0$, this momentum probe vanishes near $E=0$ and
has compact support. The Jacobian $dx=2E\,dE$ gives
$\langle U_\lambda(E^2-k^2),\Phi\rangle=\langle U_\lambda,f\rangle$.
Its inverse Fourier transform is a Schwartz spacetime probe.

For an explicit narrowing family, take $h\in C_c^\infty(-1,1)$,
$h\geq0$, $\int h=1$ and $h(0)>0$. For sufficiently small
$\varepsilon>0$, the normalized packet
\begin{equation}
 f_\varepsilon(x)=\varepsilon^{-1}
             h\bigl((x-P_\lambda)/\varepsilon\bigr),\qquad
 \int f_\varepsilon(x)\,dx=1
 \label{eq:timelike-narrow-packet}
\end{equation}
is supported in $(a,b)$. Its imaginary pole contribution is
$\pi R_\lambda h(0)/\varepsilon$. The analytic background is bounded
there; the free threshold contributes at most $O(\varepsilon^{-1/2})$
if it coincides with $P_\lambda$, and stays bounded otherwise. It
therefore cannot cancel the pole's $1/\varepsilon$ term.
Each width gives a finite answer,
but the packet height and derivatives become unbounded as the width
tends to zero. Thus this normalized family has no finite point-response
limit at $P_\lambda$.

With this deliberately unusual timelike composite-response convention,
a singular point coupling coexists with a rigorously constructed
interacting theory and finite fixed-packet answers. The convention
differs from the ordinary Euclidean amputated elementary four-point
definition used for the conjecture, so this is not a Euclidean
Landau-pole counterexample. At Euclidean squared momentum $Q^2\geq0$,
the atom instead contributes $R_\lambda/(P_\lambda+Q^2)$, and the
response is regular.

\setcounter{section}{5}
\section{A concrete one-loop asymptotically free parent}
\label{app:parent}

Can gauge and Yukawa interactions make a positive scalar quartic decrease
in the ultraviolet while every additional particle is heavy? This
appendix checks both requirements in one candidate for the $N=1$
parent of Section~\ref{sec:parent}. The retained real scalar is a gauge
singlet, so its correlations would be physical observables of a
constructed parent. The calculation supplies a tuned one-loop trajectory
and a stable massive tree vacuum; the continuum construction and the two
estimates in Theorem~\ref{thm:parent} remain further requirements.

\subsection{Fields, interactions and symmetries}

The gauge group is $SU(2)$, with generators $T^a=\sigma^a/2$ in its
fundamental representation. Write $e$ for the gauge coupling and
$D_\mu=\partial_\mu-ieT^aA^a_\mu$. Table~\ref{tab:app-parent-fields}
lists the matter content and the role of each field. In particular,
there are ten fundamental Dirac fermions in total, as well as one
Dirac singlet. Each Dirac fundamental has two gauge components.

\begin{table}[htbp]
\centering
\small
\caption{Field content of the $N=1$ parent candidate. The spectators
$R_j$ have no Yukawa couplings.}
\label{tab:app-parent-fields}
\begin{tabular}{@{}>{\raggedright\arraybackslash}p{.18\linewidth}>{\raggedright\arraybackslash}p{.24\linewidth}>{\raggedright\arraybackslash}p{.52\linewidth}@{}}
\toprule
Field & $SU(2)$ representation & Role\\
\midrule
$A^a_\mu$ & Adjoint gauge field & Three vectors, made massive by $H$.\\[3pt]
$\phi$ & Real singlet & Retained light scalar, with an exact sign symmetry.\\[3pt]
$H$ & Complex fundamental & Higgs field; gives vector masses and a heavy radial scalar.\\[3pt]
$Q_+,Q_-$ & Two Dirac fundamentals & Couple to $\phi$ through opposite Yukawa signs.\\[3pt]
$\Psi$ & One Dirac fundamental & Couples to $H$ and $\psi_s$.\\[3pt]
$R_j$, $j=1,\ldots,7$ & Seven Dirac fundamentals & Adjust the gauge beta function.\\[3pt]
$\psi_s$ & One Dirac singlet & Mixes with one component of $\Psi$ in the Higgs vacuum.\\
\bottomrule
\end{tabular}
\end{table}

Use canonical kinetic terms, real vectorlike masses
$-\sum_f m_f\bar f f$, and the following potential and Minkowski Yukawas:
\begin{equation}
\begin{aligned}
 V={}&\frac{m_0^2}{2}\phi^2-\nu^2H^\dagger H
       +\frac{\lambda_\phi}{24}\phi^4
       +\lambda_H(H^\dagger H)^2+\frac\delta2\phi^2H^\dagger H,\\
 \mathcal L_Y={}&-iy\phi(\bar Q_+\gamma_5Q_+-\bar Q_-\gamma_5Q_-)
       -i\kappa(\bar\Psi\gamma_5H\psi_s
                        +\bar\psi_s\gamma_5H^\dagger\Psi).
\end{aligned}
\label{eq:app-parent-lagrangian}
\end{equation}
$m_0^2$ is the singlet mass parameter, $\lambda_H$ the Higgs quartic,
and $\delta$ the mixed scalar coupling. Here $y,\kappa$ are real,
$\nu^2>0$ selects a Higgs vacuum, and
$\lambda_\phi/24=\lambda_\phi/4!$ fixes the singlet quartic convention.
Use Minkowski signature $(+,-,-,-)$ and
$\gamma_5=i\gamma^0\gamma^1\gamma^2\gamma^3$, with
$\gamma_5^\dagger=\gamma_5$ and $\gamma_5^2=1$.
For Dirac fields,
$(\bar f\gamma_5h)^\dagger=-\bar h\gamma_5f$.
The two displayed pseudoscalar combinations are therefore
anti-Hermitian before multiplication by $-i$, and the Yukawa
Lagrangian is Hermitian. Both left-to-right and right-to-left chiral
couplings are present. Their relative sign will matter for the fermion
masses below.

The scalar sign symmetry is
\begin{equation}
 \phi\longmapsto-\phi,\qquad Q_+\longleftrightarrow Q_-.
 \label{eq:app-parent-sign}
\end{equation}
Equal $Q_\pm$ masses preserve it, giving the required $O(1)$ symmetry
of the retained field. Assign separate conserved vector charges to
$Q_+$ and $Q_-$, exchanged by \eqref{eq:app-parent-sign}; give
$\Psi,\psi_s$ one common charge and each $R_j$ its own charge.
These assignments forbid fermion mixing between the sectors.
The further phase symmetry
$H\mapsto e^{i\theta}H$, $\psi_s\mapsto e^{-i\theta}\psi_s$
forbids replacing $H$ by the conjugate doublet
$\widetilde H=i\sigma^2H^*$ in the Higgs Yukawa.

Finally impose spatial parity combined with
$Q_+\leftrightarrow Q_-$ and $\psi_s\mapsto-\psi_s$,
taking $\phi,H$ parity even and the gauge theta angle zero.
Spatial parity reverses a pseudoscalar bilinear's sign; the additional
flavour transformations reverse it again in \eqref{eq:app-parent-lagrangian}.
This symmetry forbids the corresponding scalar Yukawas. Together
these symmetries retain precisely the Yukawa sectors displayed above.
They allow the portal $\phi^2H^\dagger H$: its vanishing on the
one-loop trajectory is a flow property, rather than a symmetry
protection.
$SU(2)$ has no local perturbative gauge anomaly for this fundamental
and singlet matter content. The vectorlike assignments permit the
gauge-invariant Dirac mass terms used above. The ten Dirac fundamentals
count as twenty left-handed Weyl doublets, an even number, so the
global $SU(2)$ gauge anomaly is absent \cite{Witten1982SU2Anomaly}.

\subsection{The one-loop flow and why its ray must be tuned}

Set
$\alpha=e^2$, $\eta=y^2$, $\zeta=\kappa^2$, and
$\dot x=dx/d\log Q$. The following one-loop beta functions use a
mass-independent subtraction scheme above the heavy thresholds.
These are Lagrangian couplings running as $Q$ varies within one
selected theory. They follow by specializing the general gauge, Yukawa and quartic
formulas of \citet[eqs.~(30), (33), (38)--(43)]{LuoWangXiao2003}
to \eqref{eq:app-parent-lagrangian}:
\begin{equation}
\begin{aligned}
16\pi^2\dot\alpha&=-\alpha^2,\\
16\pi^2\dot\eta&=22\eta^2-9\alpha\eta,\\
16\pi^2\dot\zeta&=7\zeta^2-\tfrac92\alpha\zeta,\\
16\pi^2\dot\lambda_\phi
 &=3\lambda_\phi^2+32\eta\lambda_\phi-192\eta^2+12\delta^2,\\
16\pi^2\dot\lambda_H
 &=24\lambda_H^2+(8\zeta-9\alpha)\lambda_H
                         +\tfrac98\alpha^2-4\zeta^2+\tfrac12\delta^2,\\
16\pi^2\dot\delta
 &=\delta(\lambda_\phi+12\lambda_H+4\delta
                                +16\eta+4\zeta-\tfrac92\alpha).
\end{aligned}
\label{eq:app-parent-betas}
\end{equation}
The gauge terms restrain Yukawa growth. The negative $-192\eta^2$
and $-4\zeta^2$ terms then allow positive singlet and Higgs quartics
to decrease: the $Q_\pm$ sector supplies the first term and the
$\Psi,\psi_s$ sector the second. The seven spectators change the
gauge coefficient to the small positive value $b_0=1/2$ calculated
below.
Both chiral sectors of the Higgs Yukawa contribute. For the Yukawa
interaction of $\phi$, traces over fermion components count both gauge
components of each $Q_\pm$, giving four Dirac components. The spectators contribute
only to the gauge equation at this order. In conventional notation
$16\pi^2\dot e=-b_0e^3$, their gauge count is
\begin{equation}
 b_0=\frac{11}{3}\,2-\frac43\,10\left(\frac12\right)
                    -\frac13\left(\frac12\right)=\frac12.
 \label{eq:app-parent-gauge-count}
\end{equation}
The last term counts one complex fundamental scalar. Thus
$16\pi^2\dot\alpha=-2b_0\alpha^2=-\alpha^2$, as printed above.

To look for a ray, divide all other couplings by the decreasing
gauge square:
\begin{equation}
 (r_\eta,r_\zeta,r_\phi,r_H,r_\delta)
 =\frac{(\eta,\zeta,\lambda_\phi,\lambda_H,\delta)}\alpha,
 \qquad d\tau=\frac{\alpha}{16\pi^2}\,d\log Q.
 \label{eq:app-parent-ratio-definitions}
\end{equation}
For $x=\alpha r_x$, differentiating the ratio adds $+r_x$ to
the beta polynomial divided by $\alpha^2$, because the gauge
beta function is negative. This gives the autonomous equations
\begin{equation}
\begin{aligned}
\frac{dr_\eta}{d\tau}&=r_\eta(22r_\eta-8),&
\frac{dr_\zeta}{d\tau}&=r_\zeta(7r_\zeta-\tfrac72),\\
\frac{dr_\phi}{d\tau}&=3r_\phi^2+(32r_\eta+1)r_\phi
                                  -192r_\eta^2+12r_\delta^2,\\
\frac{dr_H}{d\tau}&=24r_H^2+(8r_\zeta-8)r_H
                                   +\tfrac98-4r_\zeta^2+\tfrac12r_\delta^2,\\
\frac{dr_\delta}{d\tau}&=r_\delta(r_\phi+12r_H+4r_\delta
                                      +16r_\eta+4r_\zeta-\tfrac72).
\end{aligned}
\label{eq:app-parent-ratio-flow}
\end{equation}
The positive nonzero Yukawa roots are $r_\eta=4/11$ and
$r_\zeta=1/2$. With $r_\delta=0$, the quartic equations reduce to
\[
 3r_\phi^2+\frac{139}{11}r_\phi-\frac{3072}{121}=0,
 \qquad 24r_H^2-4r_H+\frac18=0.
\]
Choose the positive singlet root and the Higgs root $1/8$:
\begin{equation}
 r_*=\left(\frac4{11},\frac12,
          \frac{-139+\sqrt{56185}}{66},\frac18,0\right),
 \qquad r_{\phi,*}\simeq1.48536.
 \label{eq:app-parent-ray}
\end{equation}
The Higgs equation also has the root $1/24$; the chosen root gives
the tree mass pattern used below. On \eqref{eq:app-parent-ray}
every interaction is a constant times $\alpha$. At a reference
energy $Q_*$ above threshold, put $\alpha_*=\alpha(Q_*)$ and choose
$\tau(Q_*)=0$. The one-loop gauge solution is
\begin{equation}
 \alpha(Q)=\frac{\alpha_*}
 {1+\alpha_*\log(Q/Q_*)/(16\pi^2)},\qquad
 \tau(Q)-\tau(Q_*)=\log\frac{\alpha_*}{\alpha(Q)}.
 \label{eq:app-parent-gauge-solution}
\end{equation}
It tends to zero as $Q\to\infty$, so every interaction tends to
zero: the ray has complete asymptotic freedom (CAF) within this
one-loop system.

This property does not describe an open ultraviolet-attractive region
around the ray. Linearize \eqref{eq:app-parent-ratio-flow} at $r_*$.
In the displayed ratio order its Jacobian is triangular, with eigenvalues
\begin{equation}
 8,\qquad\frac72,\qquad\frac{\sqrt{56185}}{11},\qquad2,
 \qquad\frac{\sqrt{56185}+245}{66}.
 \label{eq:app-parent-eigenvalues}
\end{equation}
Each is positive. A small ratio deviation in an eigenvector direction
therefore grows as $e^{\omega\tau}=(\alpha_*/\alpha(Q))^\omega$
in the linear regime, with $\omega$ one of the positive eigenvalues
above. Approaching this ray requires
tuning the initial ratios; choosing arbitrary nearby ratios does not
establish CAF. Such tuned trajectories belong to the total
asymptotic-freedom setting of \citet{GiudiceEtAl2015}; full ultraviolet
attraction, as sought in \citet{HoldomRenZhang2015}, is a stronger
condition. Higher-loop terms must be included to find a corrected CAF
trajectory satisfying those extended RG equations. In particular, the portal may be generated:
$\delta=0$ has only been checked here at one loop.

\subsection{A stable vacuum with a heavy tree spectrum}

Let $M_H$ denote the desired heavy mass scale and evaluate the
running Lagrangian parameters at $Q_H=10M_H$, writing $\alpha_H=\alpha(Q_H)$.
The factor ten places this evaluation scale above all the tree masses
listed below.
Take the tree vacuum and parameter choices
\begin{equation}
 \phi=0,\qquad H=\begin{pmatrix}0\\v/\sqrt2\end{pmatrix},
 \qquad v=\frac{2M_H}{\sqrt{\alpha_H}},\qquad
 \nu^2=\lambda_Hv^2.
 \label{eq:app-parent-vacuum}
\end{equation}
Expanding $|D_\mu H|^2$ gives all three vector masses squared
$\alpha_Hv^2/4$. Writing the remaining radial fluctuation as
$H=(0,(v+h)/\sqrt2)^T$ gives $m_h^2=2\lambda_Hv^2$.
On $r_H=1/8$, both equal $M_H^2$. A nonzero fundamental vacuum
breaks all continuous $SU(2)$ gauge generators, so all three gauge
directions become longitudinal vector modes. The four real components
of $H$ leave only this radial scalar. The global Higgs phase creates
no additional physical Goldstone: its action on the vacuum can be
compensated by the gauge rotation
$\operatorname{diag}(e^{i\theta},e^{-i\theta})$.

Give $Q_\pm$ and the seven spectators diagonal mass $M_H$, and
give $\Psi,\psi_s$ diagonal masses $d=2M_H$. At the vacuum, the
$\Psi$ component along $H$, denoted $\Psi_\parallel$, mixes with
$\psi_s$ by $t=\kappa v/\sqrt2=M_H$, since $\zeta=\alpha_H/2$.
The orthogonal component $\Psi_\perp$ keeps its mass $2M_H$.
To see the mixed masses, put
$F=(\Psi_\parallel,\psi_s)^T$ and use
$P_{L,R}=(1\mp\gamma_5)/2$. The four-component mass operator and
the corresponding chiral mass matrix are, with $\sigma^1$ now
acting on the two entries of $F$,
\begin{equation}
 \begin{aligned}
 \mathcal L_{\rm mix}
   &=-\bar F(d\mathbf1+it\gamma_5\sigma^1)F
     =-\bar F_L\mathcal M F_R-\bar F_R\mathcal M^\dagger F_L,\\
 \mathcal M&=\begin{pmatrix}d&it\\it&d\end{pmatrix},\qquad
 \mathcal M^\dagger\mathcal M=(d^2+t^2)\mathbf1=5M_H^2\mathbf1.
 \end{aligned}
 \label{eq:app-parent-mass-matrix}
\end{equation}
Dirac masses are the singular values of $\mathcal M$, giving two
mixed modes of mass $\sqrt5M_H$. The pseudoscalar phase is why
the mixing contributes $+t^2$ to their squared masses.

\begin{table}[htbp]
\centering
\small
\caption{Tree spectrum on the chosen ray and mass assignment.
All additional modes have mass comparable to $M_H$.}
\label{tab:app-parent-masses}
\begin{tabular}{@{}p{.68\linewidth}p{.26\linewidth}@{}}
\toprule
Modes & Mass\\
\midrule
Retained singlet $\phi$ & $m>0$\\[3pt]
Three gauge vectors and radial Higgs & $M_H$\\[3pt]
$Q_\pm$ and the seven spectators & $M_H$\\[3pt]
Unmixed $\Psi_\perp$ & $2M_H$\\[3pt]
Two mixed Dirac modes from $\Psi_\parallel,\psi_s$ & $\sqrt5M_H$\\
\bottomrule
\end{tabular}
\end{table}

Classical stability and the light mass can be checked directly,
including a small portal. Put $x=H^\dagger H$ and $x_0=v^2/2$.
Using $\nu^2=2\lambda_Hx_0$, complete the square:
\begin{equation}
 \begin{aligned}
 V={}&\lambda_H\left(x-x_0+\frac{\delta\phi^2}{4\lambda_H}\right)^2
       -\lambda_Hx_0^2
       +\frac12\left(m_0^2+\frac{\delta v^2}{2}\right)\phi^2\\
    &+\frac1{24}\left(\lambda_\phi-
                         \frac{3\delta^2}{2\lambda_H}\right)\phi^4.
 \end{aligned}
 \label{eq:app-parent-square}
\end{equation}
Consequently the sufficient conditions
\begin{equation}
 \lambda_H>0,\qquad
 \lambda_\phi>\frac{3\delta^2}{2\lambda_H},\qquad
 m_0^2+\frac{\delta v^2}{2}=m^2>0
 \label{eq:app-parent-stability}
\end{equation}
make the displayed vacuum a global tree minimum and leave $\phi$
light. They hold on the ray with $\delta=0$ and $m_0^2=m^2$.
Quantum pole-mass, residue and quartic matching must replace this
classical tuning; in particular, heavy corrections of order
$gM_H^2$ to the light squared mass are absorbed into the exact mass
condition, rather than regarded as suppressed errors.

\subsection{From Lagrangian coupling ratios to physical input labels}

The ratios $r(Q)$ above compare the running Lagrangian couplings
with the running squared gauge coupling in one selected theory. Section~\ref{sec:parent} instead uses
physical input labels at the fixed subtraction scale $Q_0=\mu$:
$u_i(g)=gz_i(g)$. Thus $z_i(g)=u_i(g)/g$ compares supplied couplings
while $g$ varies through a family of theories. It is not the ratio
$r_i(Q)=c_i(Q)/\alpha(Q)$ evolving with $Q$.

The quartic normalization makes the difference concrete. For $N=1$,
\eqref{eq:bare} uses $h_B\phi^4/8$, whereas the parent uses
$\lambda_\phi\phi^4/24$. At tree level
\begin{equation}
 h_B=\frac{\lambda_\phi}{3},\qquad
 g=\frac{h_B}{32\pi^2}=\frac{\lambda_\phi}{96\pi^2}
   =c_\phi r_{\phi,*}\alpha_H,\qquad c_\phi=\frac1{96\pi^2}.
 \label{eq:app-parent-matching}
\end{equation}
For a gauge label with tree normalization $u_{\rm gauge}=e^2$,
the leading supplied ratio is therefore
$z_{\rm gauge}=1/(c_\phi r_{\phi,*})$. For example, tree conventions
\begin{equation}
 u^{\rm tree}=(e^2,y^2,\kappa^2,\lambda_H,\delta),\qquad
 z^{\rm tree}=\frac{(1,4/11,1/2,1/8,0)}{c_\phi r_{\phi,*}}
 \label{eq:app-parent-tree-labels}
\end{equation}
name the five additional coupling variables and use squares for the
gauge and Yukawa inputs. Exact physical labels still require their
specified subtraction conventions and matching corrections.
Running from $Q_H$ to $Q_0$,
threshold matching, and conversion to physical subtraction conditions
modify these tree identities. Perturbatively these corrections can
include $\alpha_H^2\ell_H$ terms, with
$\ell_H=1+\log(M_H/m_*)$ and fixed reference mass $m_*$.
Their control on the growing-mass family is part of the required
matching work.

A corrected CAF trajectory satisfying the parent's RG equations
would determine the physical $z_i(g)$ after these conversions. The ratios cannot be chosen freely
to engineer scalar coefficients. Below the heavy scale the scalar
interaction must remain weak; Section~\ref{sec:parent} checks this
scale separation for
$M_H(g)=m_*\exp[(\log(1/g))^2]$ using the scalar one-loop comparison.
Above threshold the coupled flow checked here supplies the candidate
ultraviolet behavior. Changing the probing energy holds all supplied
labels fixed; changing $g$ selects different parents with different
heavy masses.

The established ingredients are the symmetry and anomaly checks,
the six one-loop beta equations with their tuned positive ray, and
the stable massive tree spectrum. To apply Theorem~\ref{thm:parent},
one still needs a corrected CAF trajectory, a normalized continuum
parent with a stable quantum phase and physical gap, the scalar
coefficient estimate \eqref{eq:coefficient-decoupling}, and the exact
remainder estimate \eqref{eq:parent-remainder} along the same growing
mass path. Neither estimate follows from the ray or tree spectrum.
The candidate identifies these tasks for a single retained scalar;
it supplies no all-order CAF or rigorous parent construction.

\renewcommand{\bibsection}{\section*{\refname}\phantomsection\addcontentsline{toc}{section}{References}}
\setlength{\bibsep}{3pt}

\end{document}